\documentclass[journal,draftcls,onecolumn,12pt,twoside]{IEEEtran}
\usepackage[pdftex]{graphicx}
\usepackage[cmex10]{amsmath}
\usepackage{amsthm}
\usepackage{amssymb}
\usepackage{bm}
\usepackage{bbm}
\usepackage{booktabs, multirow}
\usepackage{cite}
\usepackage{color}

\newtheorem{theorem}{Theorem}
\newtheorem{corollary}[theorem]{Corollary}
\newtheorem{lemma}[theorem]{Lemma}
\newtheorem{remark}[theorem]{Remark}

\newtheorem{assumption}[theorem]{Assumption}

\newcommand{\E}{\mathbbm{E}}
\renewcommand{\Pr}{\mathbbm{P}}
\newcommand{\Cov}{\mathrm{Cov}}

\newcommand{\Var}{\mathrm{Var}}
\newcommand{\dd}{\mathrm{d}}

\newcommand{\one}{\mathbbm{1}}
\newcommand{\id}[1]{\one\{#1\}}

\newcommand{\mg}{\mathrm{M/G/1}}

\newcommand{\bg}{\mathrm{bg}}
\newcommand{\peak}{\mathrm{peak}}
\newcommand{\total}{\mathrm{total}}
\newcommand{\comb}[2]{\Biggl(\begin{array}{@{}c@{}}#1 \\ #2
\end{array}\Biggr)}

\newcommand{\st}{\mathrm{st}}
\newcommand{\icx}{\mathrm{icx}}
\newcommand{\icv}{\mathrm{icv}}
\newcommand{\cx}{\mathrm{cx}}

\usepackage{tikz}
\usetikzlibrary{automata,positioning}
\usetikzlibrary{decorations.pathreplacing,decorations.markings,shapes.geometric}
\usetikzlibrary{shapes,arrows}
\usetikzlibrary{backgrounds,calc,positioning}

\usepackage[utf8]{inputenc}
\usepackage[T1]{fontenc}
\usepackage{microtype} 
\usepackage{cite}
\usepackage{graphicx}
\usepackage{hyperref}
\usepackage{mathrsfs} 
\usepackage{subcaption}
\usepackage{amsfonts,amsmath,amsthm,amssymb}
\usepackage[a4paper,bindingoffset=0.5cm,left=2cm,right=2cm,top=2.5cm,bottom=2cm,footskip=.8cm]{geometry}
\usepackage{pgfplots}

\newcommand{\yoshiaki}[1]{#1}
\newcommand{\michel}[1]{#1}

\begin{document}
\title{
Characterizing the Age of Information with 
Multiple Coexisting Data Streams
}

\hyphenation{op-tical net-works semi-conduc-tor}

\author{
Yoshiaki Inoue and Michel Mandjes
\thanks{This work was supported in part by JSPS KAKENHI Grant Number
21H03399 and the NWO Gravitation project NETWORKS under grant agreement no. 024.002.003.}%
\thanks{Y.\ Inoue is with 
Department of Information and Communications Technology, 
Graduate School of Engineering, Osaka University, Suita 565-0871, Japan 
(e-mail: yoshiaki@comm.eng.osaka-u.ac.jp).}%
\thanks{M.\ Mandjes is with Mathematical Institute, Leiden University, P.O. Box 9512, 2300 RA Leiden, 
The Netherlands (e-mail: m.r.h.mandjes@math.leidenuniv.nl). He is also affiliated with Korteweg-de Vries Institute for Mathematics, University of Amsterdam, Amsterdam, The Netherlands; Eurandom, Eindhoven University of Technology, Eindhoven, The Netherlands; Amsterdam Business School, Faculty of Economics and Business, University of Amsterdam, Amsterdam, The Netherlands.
}%
}

\allowdisplaybreaks

\maketitle

\begin{abstract}
In this paper we analyze the distribution of the Age of Information (AoI) of a tagged data stream sharing a processor with a set of other data streams. We do so in the highly general setting in which the interarrival times pertaining to the tagged stream can have any distribution, and also the service times of both the tagged stream and the background stream are generally distributed. The packet arrival times of the background process are assumed to constitute a Poisson process, which is justified by the fact that it typically is a superposition of many relatively homogeneous streams. The first \michel{main} contribution is that we derive an expression for the Laplace-Stieltjes transform of the AoI in the resulting GI+M/GI+GI/1 model. Second, we use stochastic ordering techniques to identify tight stochastic bounds on the AoI, \michel{leading to an explicit lower and upper bound on the mean AoI}. In addition, when approximating the tagged stream's inter-generation times through a phase-type distribution (which can be done at any precision), we present a computational algorithm for the mean AoI. 
As illustrated through a sequence of numerical experiments, the analysis enables us to assess the impact of background traffic on the AoI of the tagged stream. 
\michel{It turns out that the upper bound on the mean AoI is remarkably close to its true value, which yields an explicit expression (in terms of the model parameters) for an accurate proxy of the AoI-minimizing \yoshiaki{generation} rate.}
\end{abstract}

\section{Introduction}
\label{sec:intro}

In this paper we consider the situation of an information source that
is equipped with a sensor, feeding into a processor (often referred to
as a server), and a monitor (cf.\ Figure \ref{AoI}). Due to the fact that the (typically packetized) information has to be processed by the server, it arrives with some delay at the monitor. 
The {\it Age of Information} (AoI) \cite{Kaul2011,Kaul2011-2} is a performance measure that quantifies the timeliness of the monitor’s knowledge of the information produced by the source. To conveniently analyze the AoI, queueing theory has proven to be particularly useful \cite{Kaul2012, Yates2021}.

\michel{In the theoretical analysis of the AoI, traditionally there was a strong focus on a single data stream in isolation, in that the stream under consideration is the sole user of the processor \cite{Kaul2012,Costa2016,Inoue2019,Champati2019}. Examples of papers on settings with multiple data streams are \cite{Huang2015}, analyzing a system with heterogeneous users that is modeled via a multi-class M/G/1 queue, and \cite{Sun2018}, in which multiple flows of update packets are sent over multiple servers to their respective destinations.}
A few \michel{other} exceptions, \michel{in the more recent literature}, are studies on the AoI in multi-source queueing systems with Poisson input, where inter-generation times of information packets are assumed to be exponentially distributed \cite{Yates2019,Moltafet2020,Inoue2024,Chen2022,Jiang2021,Liu2021}.
In a more typical scenario, however, a tagged data stream corresponds to packet-generation intervals with relatively low variability, sharing the resource with a massive number of competing streams. This may entail that the characteristics of these background streams may significantly impact the AoI of the tagged stream. The main objective of this paper is to quantify this effect. In our analysis we rely on the well-known result that, with the number of background streams being typically large and relatively homogeneous, the corresponding aggregate packet arrival process can be accurately approximated by a Poisson process \cite{Cinlar1972}.

\begin{figure}

\centering
\scalebox{0.99}{
\begin{tikzpicture}[scale=0.5]
\node[anchor=south west,inner sep=0] (image) at (6.4,0.7) {\includegraphics[width=0.105\textwidth]{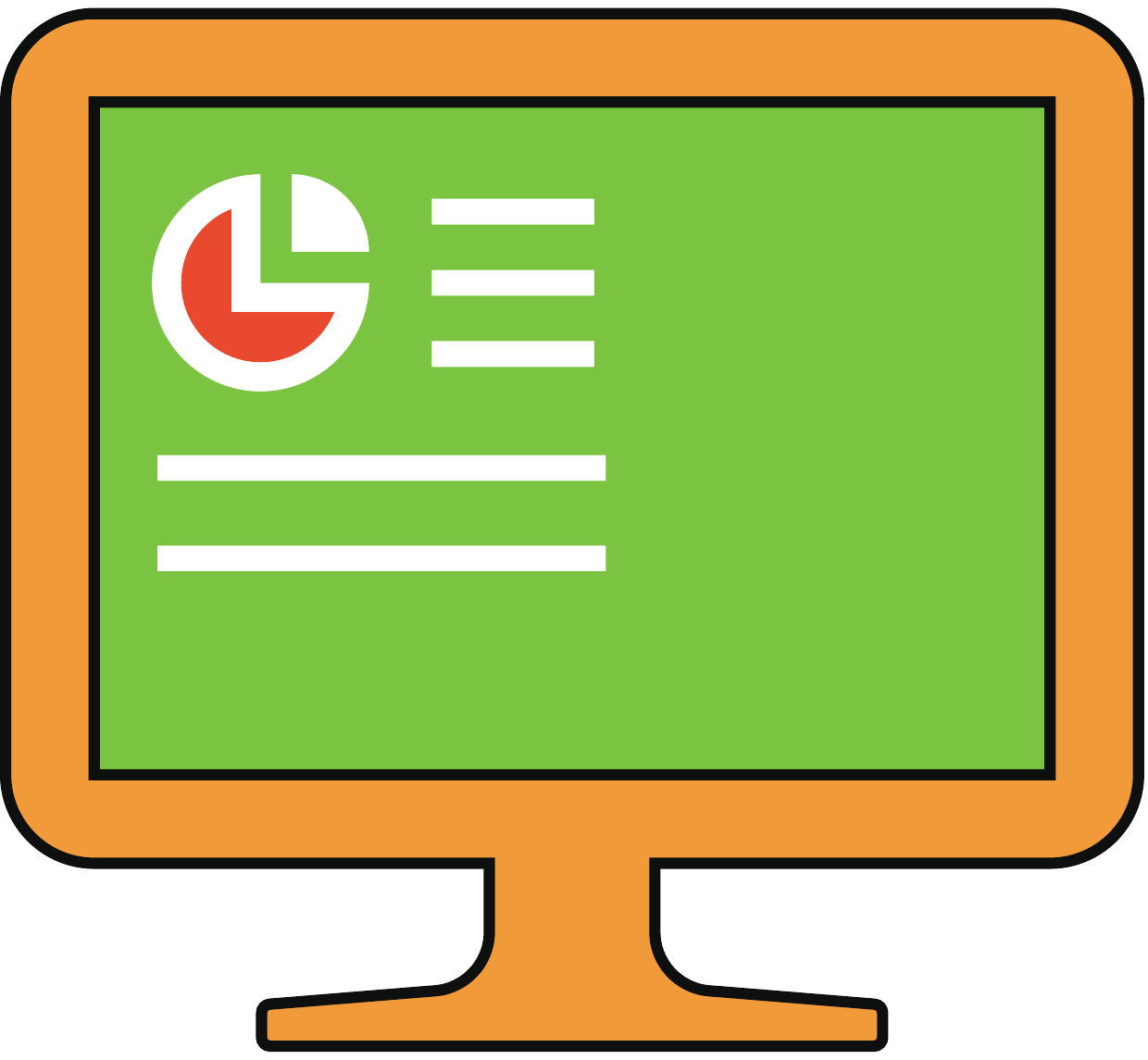}}; 
\node[anchor=south west,inner sep=0] (image) at (-4.4,0.9) {\includegraphics[width=0.129\textwidth]{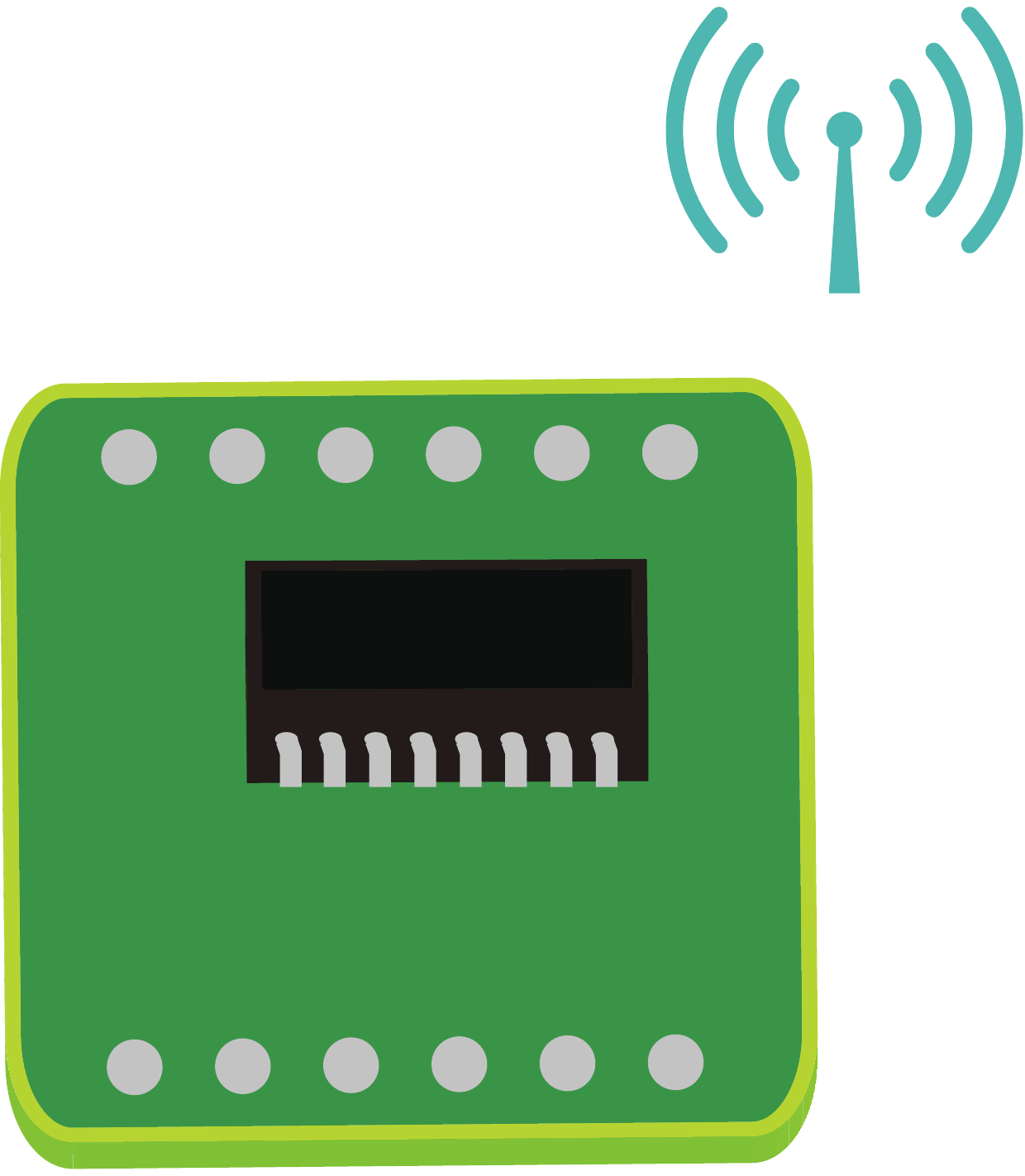}}; 
\node[anchor=south west,inner sep=0] (image) at (1.5,0.4066) {\includegraphics[width=0.08\textwidth]{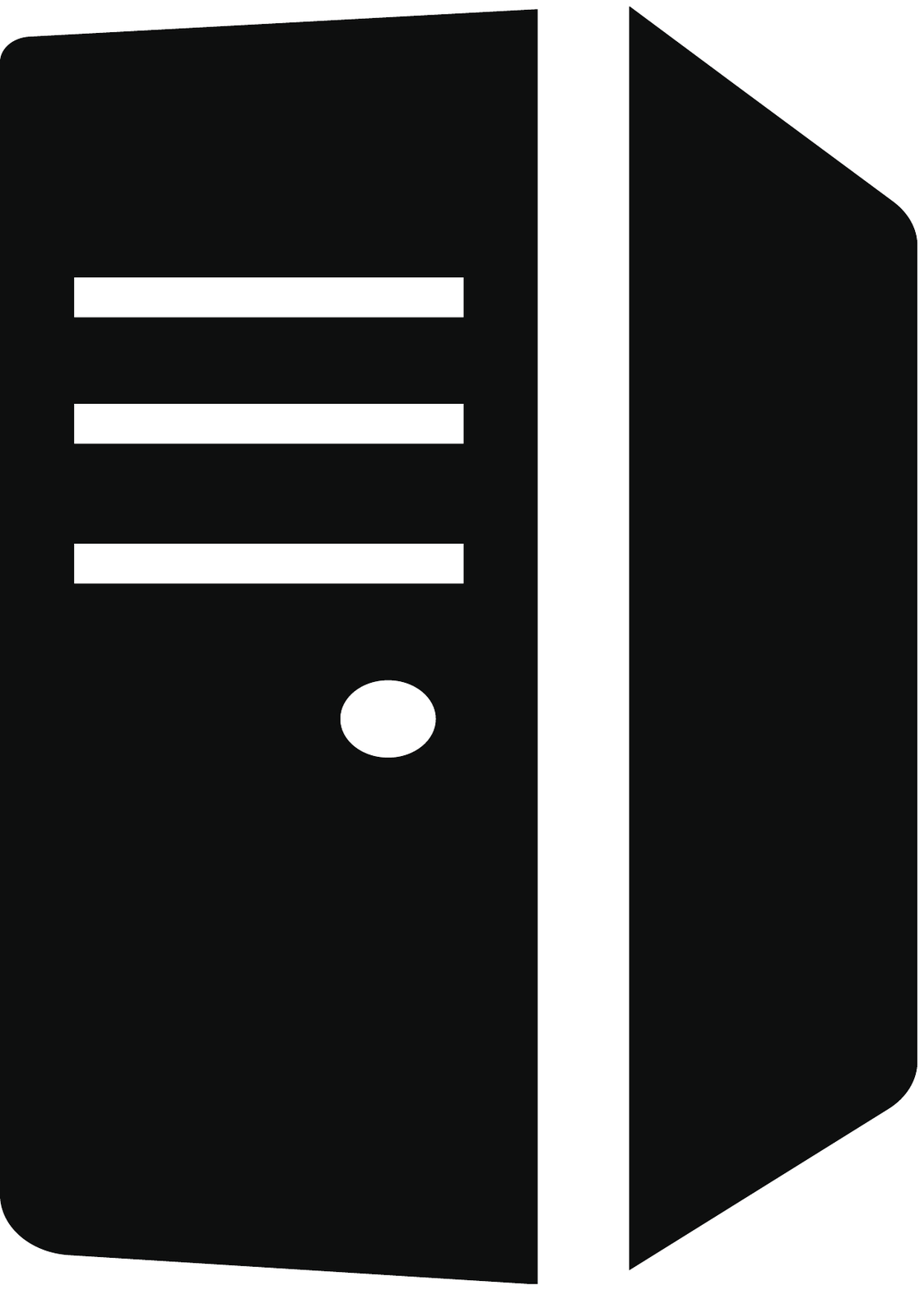}}; 
\node[anchor=south west,inner sep=0] (image) at (-1.05,1.5) {\includegraphics[width=0.08\textwidth]{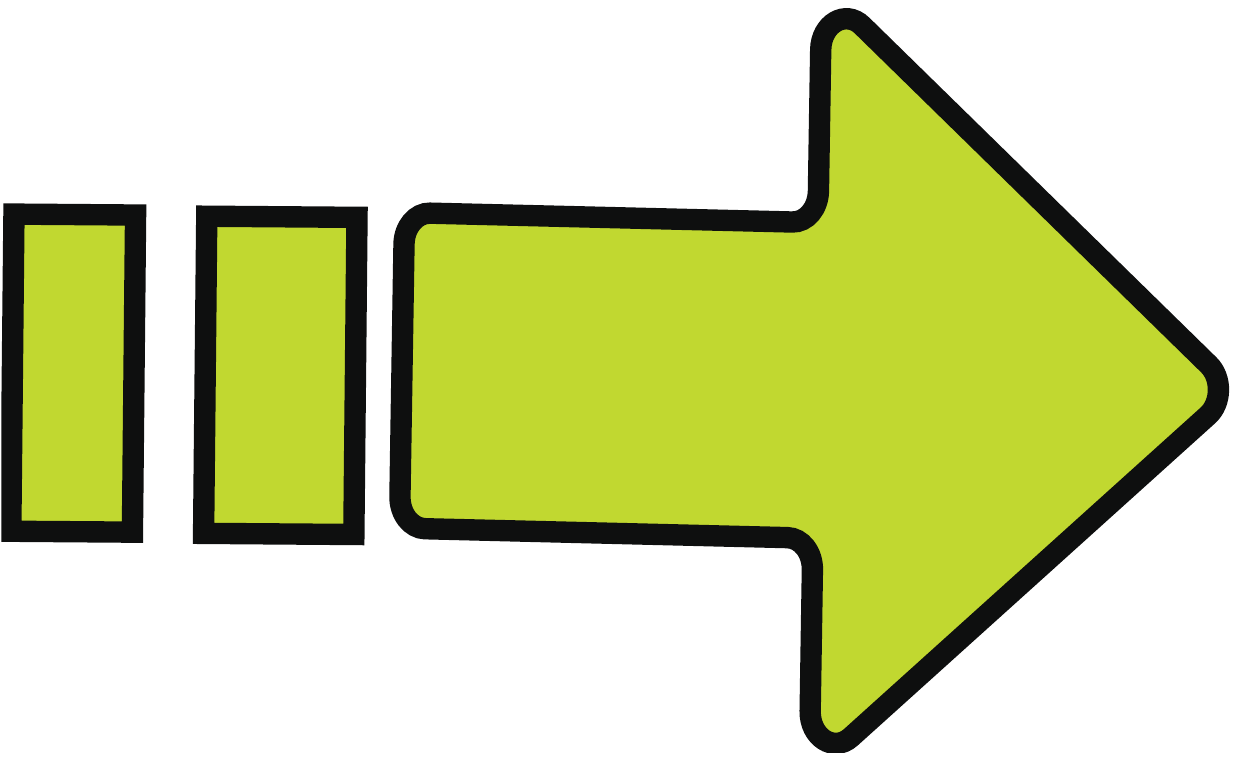}}; 
\node[anchor=south west,inner sep=0] (image) at (3.95,1.5) {\includegraphics[width=0.08\textwidth]{arrow.pdf}}; 
\node[anchor=south west,inner sep=0] (image) at (-6.96,1.5) {\includegraphics[width=0.08\textwidth]{arrow.pdf}}; 
\node[anchor=south west,inner sep=0] (image) at (-1.05,-1.1) {\includegraphics[width=0.08\textwidth]{arrow.pdf}}; 
\node[anchor=south west,inner sep=0] (image) at (3.95,-1.1) {\includegraphics[width=0.08\textwidth]{arrow.pdf}}; 

\draw (-6.9,6.3) node[below] {Information};
\draw (-6.9,5.3) node[below] {source};
\draw (-3.1,6.3) node[below] {Sensor};
\draw (2.6,6.3) node[below] {Server};
\draw (8.75,6.3) node[below] {Monitor};
\draw (0,-0.4) node[below] {Background};
\draw (0,-1.4) node[below] {stream};

\end{tikzpicture}
}

 \caption{\label{AoI}Schematic representation of the status update system, \michel{in which the information source shares the server with a background stream.}}
\end{figure}

We proceed by discussing the three main contributions of this work.  \michel{Before doing so, a general remark is that throughout this paper we rely on a variant of the compact notation system introduced by Kendall \cite[Ch.\ III.1b]{Asmussen2003}, commonly used in the AoI literature; see e.g.\ \cite{Inoue2019}. The model we focus on is for instance of GI+M/GI+GI/1 type, which means that the tagged stream has generally distributed independent (GI) interarrival times and GI service times, that the background stream is characterized by exponentially distributed inter-arrival times (M, stands for {\it memoryless}; `Poisson arrivals') and GI service times, and that there is one server.}
\begin{itemize}
    \item[$\circ$] The first part of our analysis considers the AoI of the tagged stream in the context of our \michel{general model of GI+M/GI+GI/1 type}. We succeed in evaluating the distribution of the AoI in this multi-source context through its Laplace-Stieltjes transform (LST), with the resulting expression being in terms of the (known) transform of the stationary system delay. 
    \item[$\circ$] While techniques are available to accurately invert LSTs to obtain the underlying random variable's cumulative distribution function, one would like obtain more explicit insight into the impact of the background data streams on the AoI of the tagged stream. \michel{As a second contribution, we prove a series of stochastic ordering results, by which we in particular manage to derive closed-form lower and upper bounds on the mean AoI}. 
    \item[$\circ$] Informally, the class of phase-type distributions contains convolutions and mixture of exponential distributions. Owing to the fact that any distribution on the positive half-line can be approximated arbitrarily closely by a phase-type distribution
    \cite[Th.\ III.4.2]{Asmussen2003}, we focus on the model in which the \michel{inter-generation} times of the tagged stream are of phase-type. \michel{As a third contribution, we} show how to stably compute the mean AoI, in the resulting multi-source model of the type PH+M/GI+GI/1. 
    \michel{By comparing the  resulting values with the bounds on the mean AoI, as mentioned in the second contribution, we observe that the upper bound is remarkably close. It leads to an explicit expression,  in terms of the model parameters, for an accurate proxy of the AoI-minimizing \yoshiaki{generation} rate. }
\end{itemize}

The model is more realistic than existing ones due to the fact that it allows renewal arrivals (of the tagged stream, that is), rather than just Poisson arrivals. In particular, it thus covers the practically relevant case when the tagged stream has constant interarrival times, and in addition the case that the interarrival times are `almost constant' (i.e., constant plus some perturbation). 

In previous work, all papers on multi-source models, involving a tagged stream sharing the processor with background streams, \michel{have assumed} Poisson arrivals. The consequence is that one can deal with the heterogeneity of the service times in a straightforward manner: at any arrival one can probabilistically determine whether the arrived packet corresponds to the tagged stream or to the background stream. Put differently, for any multi-source M+M/GI+GI/1 model one can work with an equivalent single source M/GI/1 model. It is remarked that in \cite{Inoue2019} the AoI under renewal arrivals (i.e., the interarrival times stemming from an \michel{arbitrary} distribution) is extensively studied, but \michel{only} in the context of a single-source model.

This paper has been organized as follows. \michel{In Section \ref{sec:literature} we provide an account of the existing literature.} Section \ref{sec:model} introduces the model, convenient notation, and some preliminaries. Then Section \ref{sec:exact} provides an exact characterization of the AoI distribution through its LST. Explicit stochastic bounds on the AoI are presented in Section \ref{sec:order}, including appealing lower and upper bounds on the mean AoI. In Section \ref{sec:PH} we specialize to the PH+M/GI+GI/1 systems, which, as argued above, is dense in the class of all GI+M/GI+GI/1 systems (in that any instance in the latter class can be approximated arbitrarily closely by an instance in the former class); we develop computational algorithms to obtain the mean AoI. In Section \ref{sec:num} we systematically assess the performance of the PH+M/GI+GI/1 system, \michel{quantifying the errors that we would face when we would have worked with the simpler M+M/GI+GI/1 system, and in addition assessing the accuracy of the bounds found in Section  \ref{sec:order}.} Section \ref{sec:conc} gives concluding remarks, including directions for potential follow-up research. 

\section{Literature overview}\label{sec:literature}

\michel{In this section we provide a brief discussion on the existing AoI structure, with emphasis on the branch that our work is in. Seen from a higher perspective, one could say there is in the first place the body of work on {\it AoI analysis}, which concerns the derivation of expressions for the AoI in various setups, often by analyzing a corresponding queueing model. The second branch could be referred to as
{\it AoI minimization}, that aims to develop techniques that can be used to minimize the AoI in specific classes of systems. This typically boils down to analyzing scheduling problems to optimize the generation rate and resource allocation (i.e., the the order of service in the underlying queueing system). The main tools used are Markov decision processes and, more generally, techniques from optimization. Importantly, this categorization is not strict: our work for instance qualifies as AoI analysis, but we also provide explicit guidelines for setting the optimal generation rate in the setup considered.}

\smallskip

\michel{This paragraph presents a (non-exhaustive) account of papers on AoI analysis in the multi-source context, primarily focusing on papers in which the underlying queueing discipline is first-come first-served (FCFS). 
The pioneering paper \cite{Yates2019} quantifies the AoI of a tagged stream interfering with a background stream, but in a setting in which interarrival times as well as service times are exponentially distributed. For the resulting M+M/M+M/1 model, and actually also for the analogous model in which more than two streams share the queueing resource, the authors 
analyze the mean of the AoI of each of the streams. This analysis is extended in \cite{Moltafet2020,Inoue2024} to general service times (but still Poisson arrivals). The fact that it is reasonable to drop older packets when new packets arrive at the queue, justifies the interest in models with finite buffer capacity. In the single-stream case, it is natural to consider the case that the buffer can contain precisely one packet. The corresponding multi-source model (with a buffer size of one, that is) is dealt with in \cite{Chen2022}, but is remarked that in this setting with multiple coexisting sources it would be more appropriate to work with a setup in which {\it all} interfering streams are given a dedicated buffer of size one.
In \cite{Jiang2021}, relying on Palm calculus techniques, an expression for the joint transform of the AoI\,s is obtained, which is valid under fairly general conditions. It should be noted, however, that its application to specific situations requires additional non-trivial analysis, and so far such analysis has succeeded only in special cases (namely, the input is Poisson and the service discipline corresponding to preemptive last-come first-served). We finally mention \cite{Liu2021}, in which the uncertainty in the arrival and service processes is captured through the use of uncertainty sets. The analysis leads to bounds on the 
peak AoI (defined as the largest value of the AoI, before it drops due to a new update), accurate under both light and heavy loads.  This literature review confirms that, before our work, no results were available for the AoI distribution in systems of the practically relevant GI+M/GI+GI/1 type.}

\smallskip

\michel{We conclude this section by discussing a set of representative recent papers on AoI minimization in the multi-source context. We do not aim to provide a comprehensive overview; instead, the selected papers are meant to highlight the most prominent research themes within this area.
An early paper on AoI minimization is \cite{He2018}, that focuses on designing a policy that delivers all messages at the end-users as fresh as possible. In the wireless network setting considered, each timeslot it has to be decided which links (i.e., transmitter-receiver pairs) are to be activated so as to minimize the overall AoI. The authors cast this scheduling problem in terms of an optimization problem, and analyze its complexity. A similar setup is studied in \cite{Kadota2018, Kadota2019}, but then aiming at minimizing the weighted sum of the links' expected AoI\,s, with the additional feature of the links being unreliable (i.e., packet transmission fails with some probability). A main contribution of \cite{Kadota2018, Kadota2019} lies in the assessment of {\em low-complexity} scheduling algorithms, including one of max-weight type and one based on the Whittle index. 
The minimization of a general AoI-penalty function in a multi-server system is investigated in \cite{Sun2018}, assuming synchronized packet arrivals.
A novel element in \cite{Yin2019} is the introduction of the concept of {\it effective AoI} to represent the AoI specifically at time instants of user requests. The link activation problem discussed above is dealt with, relying on a decomposition into multiple computationally tractable subproblems and the application of restless multi-armed bandit techniques. In \cite{Hsu2020} the focus is on a setting in which the information packets arrive in a random fashion, covering an approach relying on Whittle’s methodology for restless bandits, but also exploring an index scheduling algorithm in which the decision problem is stated in terms of a Markov decision process. In the framework studied in \cite{Bedewy2021} there is, besides the scheduler, a {\it sampler} that serves to determine when a source should generate a new packet. The paper develops a dynamic programming framework for optimal scheduling/sampling strategies. In \cite{Hirosawa2020} the focus is on AoI in the context of energy harvesting wireless sensor networks, the main contribution being resource allocation algorithms that minimize the average AoI. Some recent papers consider more refined objective functions; see e.g.\ \cite{Li2022} where the goal is to minimize maximum AoI thresholds, and \cite{Maatouk2022} which considers a setting with different priority levels of sensors. An example of a paper in which reinforcement learning is used is \cite{Zakeri2024}: it states the AoI minimization as a stochastic control optimization problem and solves it by using a constrained Markov decision process, but it in addition develops a model-free deep reinforcement learning policy that performs the scheduling decisions dynamically.}

\section{Model, notation, and preliminaries}\label{sec:model}

Throughout this paper, we follow a convention that for a non-negative
random variable $Y$, its cumulative distribution function (CDF) and
probability density function (if exists) are denoted by $F_Y(\cdot)$
and $f_Y(\cdot)$. We also define $\overline{F}_Y(\cdot)$ as
its complementary CDF and $f_Y^*(\cdot)$ as its LST:
\begin{align*}
\overline{F}_Y(x) = 1 - F_Y(x),
\;\;
x \geq 0,
\qquad
f_Y^*(s) 
= 
\E[e^{-sY}] = \int_0^{\infty} e^{-sx} \dd F_Y(x),
\;\;
\michel{s \in {\mathbb C}}.
\end{align*}

\michel{\subsection{Sensor and monitor}}
Suppose that a sensor observes (i.e., samples) the time-varying state of an
information source at a rate $\lambda > 0$.
At the sampling time instants, the sensor generates an information
packet containing the obtained data and transmits this to a remote monitor. 
Each information packet experiences some delay at an
intermediate communication channel or a processor which performs 
a computation to extract information from the raw data.
More specifically, let $(\alpha_n)_{n=0,1,\ldots}$ denote the
sequence of generation times of information packets, where 
$\alpha_{n-1} \leq \alpha_n$ ($n = 1,2,\ldots$).
We refer to the information packet generated at time $t = \alpha_n$ as
the $n$th packet. By definition, the inter-generation time $G_n$
between the $(n-1)$th and $n$th packets is given by
\begin{equation}
G_n = \alpha_n - \alpha_{n-1},
\quad
n = 1,2,\ldots.
\label{eq:G_n-def}
\end{equation}
The $n$th packet (more precisely, the information it contains) is
received by the monitor at time $\beta_n$, i.e., the system delay
$D_n$ experienced by the $n$th packet is given by
\begin{equation}
D_n = \beta_n - \alpha_n,
\quad
n = 0,1,\ldots.
\label{eq:D_n-def}
\end{equation}
We hereafter assume that the time axis is taken so 
that
%
\begin{equation}
\alpha_0 > 0,
\label{eq:alpha-condition}
\end{equation}
which is convenient in defining the AoI, as we will see later.

\begin{figure}

   \centering
\begin{tikzpicture}[scale=0.84]
\draw[thick] (-1,0)--(7,0);
\draw[thick] (-1,2)--(7,2);
\draw[thick] (7,2)--(7,0);

\draw(6.9,0.1)--(6.9,1.9);
\draw(6.0,0.1)--(6.9,0.1);
\draw(6.0,1.9)--(6.9,1.9);
\draw(6.0,0.1)--(6.0,1.9);

\filldraw[fill=black!40!white, draw=black] (6.0,0.1) rectangle (6.9,1.9);
\filldraw[fill=black!40!white, draw=black] (5.4,0.1) rectangle (5.9,1.9);
\filldraw[fill=black!40!white, draw=black] (0.7,0.1) rectangle (1.0,1.9);
\filldraw[fill=black!10!white, draw=black](8,1) circle (0.9);

\draw(5.9,0.1)--(5.4,0.1);
\draw(5.9,0.1)--(5.9,1.9);
\draw(5.9,1.9)--(5.4,1.9);
\draw(5.4,0.1)--(5.4,1.9);

\draw(5.3,0.1)--(5.3,1.9);
\draw(4.0,0.1)--(5.3,0.1);
\draw(4.0,1.9)--(5.3,1.9);
\draw(4.0,0.1)--(4.0,1.9);

\draw(3.9,0.1)--(3.9,1.9);
\draw(3.0,0.1)--(3.9,0.1);
\draw(3.0,1.9)--(3.9,1.9);
\draw(3.0,0.1)--(3.0,1.9);

\draw(2.9,0.1)--(2.9,1.9);
\draw(1.1,0.1)--(2.9,0.1);
\draw(1.1,1.9)--(2.9,1.9);
\draw(1.1,0.1)--(1.1,1.9);

\draw(1.0,0.1)--(1.0,1.9);
\draw(0.7,0.1)--(1.0,0.1);
\draw(0.7,1.9)--(1.0,1.9);
\draw(0.7,0.1)--(0.7,1.9);

\draw(0.0,0.1)--(0.0,1.9);
\draw(0.6,0.1)--(0.0,0.1);
\draw(0.6,1.9)--(0.0,1.9);
\draw(0.6,0.1)--(0.6,1.9);

\draw[->, thick](-4,0.4)--(-1.2,0.4);

\draw[->, thick](-4,1.6)--(-1.2,1.6);

\draw (8,1) node[centered] {$\mu$};
\draw (-3,1.6) node[above] {$X(t)$};
\draw (-3,0.4) node[below] {$X_{\rm bg}(t)$};
\end{tikzpicture}
    \caption{\label{figMG1}Schematic representation of queue with coexisting data streams. The grey blocks correspond to packets in the tagged stream $X(t)$, and the white blocks to packets in the background stream $X_{\rm bg}(t)$.}
\end{figure}

\michel{\subsection{Queueing Representation}}
In this paper, we model the sequence of system delays
$(D_n)_{n = 0,1,\ldots}$ via the use of a single-server queue.
This queue is fed by a \emph{tagged input}, corresponding to 
the sensor-monitor pair we focus on, as well as 
\emph{additional exogenous inputs}, which represent the
background traffic load generated by sources other than the
tagged input (cf.\ Figure \ref{figMG1}). More specifically, we
consider a queueing system in which the arrival process consists of 
two different input streams, in the sequel denoted
by $(X(t))_{t \geq 0}$ and $(X_{\bg}(t))_{t \geq 0}$,
where $X(t)$ and $X_{\bg}(t)$ denote the cumulative amount of work
brought into the system by the sensor and the background traffic, respectively.
The workload in system $V(t)$ at time $t$ is then given by
\begin{equation}
V(t) = V(0) + X(t) + X_{\bg}(t) - \mu t + \mu \int_{u \in [0,t)}
\id{V(u) = 0} \,\dd u,
\quad
t \geq 0,
\label{eq:V(t)-def}
\end{equation}
where $\mu> 0$ denotes the service rate. 
See Figure \ref{fig:X-V-A} for a graphical explanation. 
Assuming that the input processes $(X(t))_{t \geq 0}$ and
$(X_{\bg}(t))_{t \geq 0}$ are {\it c\`adl\`ag} (i.e., right continuous with left limits),
so is the associated workload process $(V(t))_{t \geq 0}$.

\begin{figure}
\centering
\resizebox{14cm}{2.5cm}{
 \pgfplotstableread{Data_X.txt}{\table}
    \begin{tikzpicture}
        \begin{axis}[
            xmin = 0, xmax = 200,
            ymin = 0, ymax = 3,
            xtick distance = 20,
            ytick distance = 0.5,
            grid = both,
            minor tick num = 1,
            major grid style = {lightgray},
            minor grid style = {lightgray!25},
            width = \textwidth,
            height = 0.24\textwidth,
            legend cell align = {left},
            legend pos = north west
        ]
            \addplot[blue, mark = *, mark size = 0.3pt, line width = 1pt] table [x = {x}, y = {a}] {\table};
                       \legend{$X(t)$}
        \end{axis}
\end{tikzpicture}}

\smallskip

\resizebox{14cm}{2.5cm}{
 \pgfplotstableread{Data_Xbg.txt}{\table}
    \begin{tikzpicture}
        \begin{axis}[
            xmin = 0, xmax = 200,
            ymin = 0, ymax = 3,
            xtick distance = 20,
            ytick distance = 0.5,
            grid = both,
            minor tick num = 1,
            major grid style = {lightgray},
            minor grid style = {lightgray!25},
            width = \textwidth,
            height = 0.24\textwidth,
            legend cell align = {left},
            legend pos = north west
        ]
            \addplot[blue, mark = *, mark size = 0.3pt, line width = 1pt] table [x = {x}, y = {a}] {\table};
                       \legend{$X_{\rm bg}(t)$}
        \end{axis}
    \end{tikzpicture}}

\smallskip
    
\resizebox{14cm}{2.5cm}{
 \pgfplotstableread{Data_V.txt}{\table}
    \begin{tikzpicture}
        \begin{axis}[
            xmin = 0, xmax = 200,
            ymin = 0, ymax = 35,
            xtick distance = 20,
            ytick distance = 10,
            grid = both,
            minor tick num = 1,
            major grid style = {lightgray},
            minor grid style = {lightgray!25},
            width = \textwidth,
            height = 0.24\textwidth,
            legend cell align = {left},
            legend pos = north west
        ]
            \addplot[blue, mark = *, mark size = 0.3pt, line width = 1pt] table [x = {x}, y = {a}] {\table};
                       \legend{$V(t)$}
        \end{axis}
    \end{tikzpicture}}

\smallskip
    
\resizebox{14cm}{2.5cm}{
\pgfplotstableread{Data_A.txt}{\table}
  \begin{tikzpicture}
      \begin{axis}[
          xmin = 0, xmax = 200,
          ymin = 0, ymax = 35,
          xtick distance = 20,
          ytick distance = 10,
          grid = both,
          minor tick num = 1,
          major grid style = {lightgray},
          minor grid style = {lightgray!25},
          width = \textwidth,
          height = 0.24\textwidth,
          legend cell align = {left},
          legend pos = north west
      ]
          \addplot[blue, mark = *, mark size = 0.3pt, line width = 1pt] table [x = {x}, y = {a}] {\table};
                     \legend{$A(t)$}
\end{axis}
\end{tikzpicture}}

\caption{
The top graphs represent sample paths of the arrival times and service requirements corresponding to the processes $X(t)$ and $X_{\rm bg}(t)$. The bottom graphs represent the resulting workload process $V(t)$ and AoI process $A(t)$. The service rate $\mu$ is equal to 1.
The tagged stream has interarrival times that are Erlang with 10 phases and mean 20, whereas the service requirements are Erlang with 5 phases and mean 1. The background stream is Poissonian with arrival rate $0.9$, whereas the service requirements are Erlang with 5 phases and mean 1.}
\label{fig:X-V-A}
\end{figure}

Letting $H_n$ ($n = 0,1,\ldots$) denote the service requirement of the
$n$th packet generated by the sensor, $X(t)$ can be represented as
\[
X(t) = \sum_{n=1}^{\infty} \id{0 \leq \alpha_n \leq t} H_n. 
\]
In this paper, we let the inter-generation times pertaining to the sensor be denoted by
$(G_n)_{n=1,2,\ldots}$, and the corresponding service requirements by
$(H_n)_{n=0,1,\ldots}$. We assume that these are two sequences of independent and identically distributed
(i.i.d.) random variables, with the respective CDFs $F_G(\cdot)$ and $F_H(\cdot)$, where its is in addition assumed that both sequences are independent of each other.
On the other hand, we assume that the background traffic
$(X_{\bg}(t))_{t \geq 0}$ is a compound Poisson process: it has
upward jumps following a Poisson process with rate
$\lambda_{\bg}$, where the jump sizes are i.i.d.\ with CDF $F_{H_{\bg}}(\cdot)$.
The background traffic load is then characterized as
\[
\Pr(X_{\bg}(t) \leq x) 
=
\sum_{k=0}^{\infty}
\frac{e^{-\lambda_{\bg}t}(\lambda_{\bg} t)^k}{k!}
\cdot
F_{H_{\bg}}^{\star k}(x),
\quad
x \geq 0,
\]
where $F_{H_{\bg}}^{\star k}(\cdot)$ denotes the $k$-fold convolution of
$F_{H_{\bg}}(\cdot)$, being defined recursively via
\[
F_{H_{\bg}}^{\star 1}(x) = F_{H_{\bg}}(x),
\quad
F_{H_{\bg}}^{\star k}(x) 
= 
\int_0^x F_{H_{\bg}}^{\star k-1}(x-y)\, \dd F_{H_{\bg}}(y),
\;\;
k = 2,3,\ldots.
\]
The processes $(X(t))_{t \geq 0}$ and $(X_{\bg}(t))_{t \geq 0}$ are assumed independent.
The resulting queueing model could be denoted by GI+M/GI+GI/1, where the +
symbol stands for `superposition': a renewal input stream with generally distributed service times shares a single-server queueing resource with a Poisson input stream with generally distributed service times. As mentioned in the introduction, the Poissonian nature of the background arrivals is motivated by the fact that the number of background streams is typically large and fairly homogeneous. 

\subsection{\michel{Age of Information}}
\michel{We have now introduced all concepts needed to formally define the Age of Information. 
Observe that, in terms of the workload process,} the
system delay $D_n$ of the $n$th packet is given by
\begin{align}
D_n = \frac{V(\alpha_n-) + H_n}{\mu} 
= 
\frac{V(\alpha_n)}{\mu},
\label{eq:D_n-by-V}
\end{align}
where the second equation follows from the right-continuity of
$(V(t))_{t \geq 0}$. The AoI $A(t)$ at time $t$ is
defined as the elapsed time since the generation time of the last
packet received by the monitor (cf. Figure \ref{fig:X-V-A}):
\begin{equation}
A(t) = 
\left\{
\begin{array}{@{}l@{\quad}l}
t - A(0), & 0 \leq t < \beta_0,
\\[1ex]
t - \max\{\alpha_n;\; \alpha_n + D_n \leq t\},
&
t \geq \beta_0.
\label{eq:A(t)-def}
\end{array}
\right.
\end{equation}
Recall that $\alpha_n + D_n = \beta_n$ represents the reception time
of the $n$th packet, so that $t \geq \beta_0$ in the second case of 
(\ref{eq:A(t)-def}) is required for the set $\{\alpha_n;\; \alpha_n + D_n \leq t\}$ to be
non-empty. Also, note that the condition (\ref{eq:alpha-condition}) ensures 
$\beta_0 > 0$. Let $A_{\peak,n}$ ($n=1,2,\ldots$) denote the $n$th peak AoI, i.e.,
the value of the AoI just before the $n$th packet is received by the
monitor:
\begin{align}
A_{\peak,n} 
= 
\lim_{t \to \beta_n-} A(t) 
= 
\beta_n - \alpha_{n-1}
&=
\beta_n - \alpha_n + \alpha_n - \alpha_{n-1}
\nonumber
\\
&=
D_n + G_n,
\label{eq:A_peak-DG}
\end{align}
where the last equation follows from (\ref{eq:G_n-def}) 
and (\ref{eq:D_n-def}).

\subsection{\michel{Stationary Behavior}}
Let $G$, $H$, and $H_{\bg}$ respectively denote generic random
variables following the CDFs $F_G(\cdot)$, $F_H(\cdot)$, and
$F_{H_{\bg}}(\cdot)$. Also recall that $\lambda := 1/\E[G]$ denotes the
\yoshiaki{generation} rate of the sensor. We define the traffic intensity of
respective streams as
\[
\rho := \frac{\lambda\E[H]}{\mu},
\quad
\rho_{\bg} := \frac{\lambda_{\bg}\E[H_{\bg}]}{\mu}.
\]
In the rest of this paper, we assume that the stability condition
\begin{equation}
\rho + \rho_{\bg} < 1,
\label{eq:stability}
\end{equation}
is in place. Because the system is regenerative and the regeneration
time is almost surely (a.s.) finite under (\ref{eq:stability}), the time-averaged
distributions of $(D_n)_{n = 0,1,\ldots}$,
$(A_{\peak,n})_{n=1,2,\ldots}$, and $(A(t))_{t \geq 0}$ agree with 
their respective stationary distributions $F_D(\cdot)$, $F_{A_{\peak}}(\cdot)$,
and $F_A(\cdot)$, i.e., the following relations hold a.s.:
\begin{align}
\lim_{N \to \infty}\frac{1}{N}\sum_{n=0}^{N-1} \id{D_n \leq x}
&=
F_D(x),
\qquad
x \geq 0,
\nonumber
\\
\lim_{N \to \infty}\frac{1}{N}\sum_{n=1}^{N} \id{A_{\peak,n} \leq x}
&=
F_{A_{\peak}}(x), 
\qquad
x \geq 0,
\nonumber
\\
\lim_{T \to \infty}\frac{1}{T}\int_0^T \id{A(t) \leq x} \dd t
&=
F_A(x),
\qquad
x \geq 0.
\label{eq:A-time-avg-dist}
\end{align}
We thus define $D$, $A_{\peak}$, and $A$ as 
generic random variables corresponding to the stationary versions of $D_n$, $A_{\peak,n}$, and $A(t)$,
respectively.

\section{Exact Formula for the AoI Distribution}\label{sec:exact}

\michel{In this section we derive, in Theorem \ref{theorem:A-by-D}, an exact expression for the LST of the
stationary AoI~$A$. Importantly, we allow the tagged stream, corresponding to
the sensor-monitor pair, to have a general \michel{inter-generation} time distribution.}
\michel{Having an expression for the LST $f^*_{A}(s)$ is particularly useful in various respects: 
\begin{itemize}
    \item[$\circ$] In the first place one can use computational software to invert it (for instance based on the algorithm proposed in \cite{Abate1992}), so as to obtain a numerical approximation of the density of $A$. Also, when inverting $f^*_{A}(s)/s$ one obtains a numerical approximation of the corresponding cumulative distribution function. 
    \item[$\circ$] In the second place, the expression in Theorem~\ref{theorem:A-by-D} reveals that, distributionally, the stationary AoI $A$ can be written as the sum of $H/\mu$ and a second random variable. This second random variable has an LST that can be written as the expectation of the product of two LSTs, thus yielding an appealing interpretation; see Remark \ref{remark:A-by-D}. 
    \item[$\circ$]  In the third place, by differentiating the LST (with respect to $s$, that is) and inserting $s=0$ provides us with the moments of the stationary AoI; we will use this property in Section \ref{sec:PH}.  
\end{itemize}}

We utilize the following characterization of the stationary
distribution of the AoI \cite{Inoue2019}:
\begin{align}
F_A(x) &= 
\frac{1}{\E[G]}
\int_0^x 
\left\{\overline{F}_{A_{\peak}}(y) - \overline{F}_D(y)\right\}\dd y,
\quad
x \geq 0,
\label{eq:F_A-formula}
\\
f_A^*(s) &= \frac{f_D^*(s) - f_{A_{\peak}^*}(s)}{s\E[G]},
\quad
\michel{s \in{\mathbb C}}.
\label{eq:f_A*-formula}
\end{align}
Note here that (\ref{eq:A_peak-DG}) implies $\E[A_{\peak}] - \E[D] =
\E[G]$, with which we can verify that $F_A(\cdot)$ and 
$f_A^*(\cdot)$ given by (\ref{eq:F_A-formula}) and (\ref{eq:f_A*-formula})
satisfy $\lim_{x \to \infty} F_A(x) = 1$ and $\lim_{s \downarrow 0} f_A^*(s)
= 1$. By (\ref{eq:F_A-formula}) and (\ref{eq:f_A*-formula}), the analysis
of the stationary AoI $A$ is reduced to that of the stationary system
delay $D$ and the stationary peak AoI $A_{\peak}$. 

Let $B_{\bg}$ denote the length of a busy period of
the M/G/1 queue only with the background traffic load; 
it is a standard result from queueing theory that its 
LST $f_{B_{\bg}}^*(s)$ satisfies the {\it Kendall functional equation} \cite[Eq.\ (1)]{Abate1995}
\begin{equation}
f_{B_{\bg}}^*(s) 
= 
f_{H_{\bg}}^*\left( 
\frac{s + \lambda_{\bg} - \lambda_{\bg}f_{B_{\bg}}^*(s)}{\mu}
\right).
\label{eq:f_B_bg-def}
\end{equation}
It is easily seen that for each $s \geq 0$, 
this equation and the condition $|f_{B_{\bg}}^*(s)| \leq 1$ 
uniquely determine the value of $f_{B_{\bg}}^*(s)$; to this end, observe that 
$f_{B_{\bg}}^*(s)$ is the solution of the equation $y=f^*(y)$ with a function
\[f^*(y)=f_{H_{\bg}}^*\left( 
\frac{s + \lambda_{\bg} - \lambda_{\bg}y}{\mu}
\right), \]
which is a convex function satisfying $f^*(0)>0$ and $f^*(1)= f_{H_{\bg}}^*(s/\mu)<1$.
\begin{lemma}
\label{lemma:A_peak-by-D}

The LST of the stationary peak AoI $f_{A_{\peak}}^*(\cdot)$ is given in
terms of the LST of the stationary system delay $f_D^*(\cdot)$ by
\begin{equation}
f_{A_{\peak}}^*(s) 
= 
f_H^*(s/\mu)
\int_{y=0}^{\infty}
e^{-sy}
\mathcal{L}_{\omega}^{-1}
\left[
\frac{1}{\phi(s) - \omega}
\left(
\frac{s}{\psi(\omega)}
\cdot
f_D^*(\psi(\omega)) - f_D^*(s)
\right)
\right]\! (y)
\; \dd F_G(y),
\label{eq:A_peak-by-D}
\end{equation}
where $\phi(s)$ and $\psi(\omega)$ are defined as
\begin{align}
\phi(s) &:= s - \lambda_{\bg} + \lambda_{\bg}f_{H_{\bg}}^*(s/\mu),
\label{eq:phi-def}
\\
\psi(\omega) &:= \omega + \lambda_{\bg} - \lambda_{\bg}f_{B_{\bg}}^*(\omega),
\label{eq:psi-def}
\end{align}
and $\mathcal{L}_{\omega}^{-1}[f(s,\omega)]$ denotes the inverse
Laplace transform of $f(s,\omega)$ with respect to $\omega$:
\[
f(s,\omega)
=
\int_0^{\infty} \michel{\mathcal{L}_\omega^{-1}}[f(s,\omega)](t) e^{-\omega t} \dd t.
\]
\end{lemma}
\begin{proof}
From (\ref{eq:D_n-by-V}) and (\ref{eq:A_peak-DG}), we have 
\begin{align}
A_{\peak,n} &= \frac{V(\alpha_n-) + H_n}{\mu} + G_n
= 
\hat{V}(\alpha_n-)+ G_n
+ \frac{H_n}{\mu},
\label{eq:A_peak-by-VHG}
\end{align}
where $\hat{V}(t)$ denotes the virtual waiting time:
\begin{equation}
\hat{V}(t) = \frac{V(t)}{\mu}.
\label{eq:hatV-def}
\end{equation}
Because $(X(t))_{t \geq 0}$ has no jumps in $t \in [\alpha_{n-1},
\alpha_n)$, in this interval (\ref{eq:V(t)-def}) can be rewritten as
\[
V(t) = V(\alpha_{n-1}) + X_{\bg}(t) - \mu t 
+ \mu \int_{u \in [\alpha_{n-1},t)}
\id{V(u) = 0} \,\dd u,
\quad
t \in [\alpha_{n-1}, \alpha_n),
\]
i.e., $(V(t))_{t \in [\alpha_{n-1}, \alpha_n)}$ behaves as if it is the workload process of an ordinary M/GI/1
queue with the arrival rate $\lambda_{\bg}$, the service requirement distribution
$F_{H_{\bg}}(\cdot)$, and the service rate $\mu$. 
Let $(V_{\mg}(t))_{t \geq 0}$ denote the workload process of such an
M/G/1 queue and let $\hat{V}_{\mg}(t) = V_{\mg}(t)/\mu$ denote the
corresponding virtual waiting time. It is known that
the LST of $\hat{V}_{\mg}(t)$ satisfies \cite[P.\ 83]{Takagi1991}
\begin{align}
\int_{t=0}^{\infty} 
\E\bigl[e^{-s\hat{V}_{\mg}(t)}\bigr]
e^{-\omega t} 
\dd t
&=
\frac{1}{\phi(s) - \omega}
\left(
\frac{s}{\psi(\omega)}
\cdot
f_{\hat{V}_{\mg}(0)}^*(\psi(\omega))
- f_{\hat{V}_{\mg}(0)}^*(s)
\right).
\label{eq:hatV-transient}
\end{align}
Therefore, we have from (\ref{eq:A_peak-by-VHG}),
\begin{align*}
\E[e^{-sA_{\peak,n}}]
&=
f_H^*(s/\mu)
\cdot
\E[e^{-s(\hat{V}(\alpha_n-) + G_n)}]
\\
&=
f_H^*(s/\mu)
\int_{x=0}^{\infty}
\E[e^{-s(\hat{V}(\alpha_n-) + G_n)} \mid \hat{V}(\alpha_{n-1}) = x]
\,\dd F_{\hat{V}(\alpha_{n-1})}(x)
\\
&=
f_H^*(s/\mu)
\int_{x=0}^{\infty}
\E[e^{-s(\hat{V}_{\mg}(G_n) + G_n)} \mid \hat{V}_{\mg}(0) = x]
\,\dd F_{D_{n-1}}(x)
\\
&=
f_H^*(s/\mu)
\int_{x=0}^{\infty}
\dd F_{D_{n-1}}(x)
\int_{y=0}^{\infty}
e^{-sy}\,
\E[e^{-s\hat{V}_{\mg}(y)} \mid \hat{V}_{\mg}(0) = x]
\,\dd F_G(y)
\\
&=
f_H^*(s/\mu)
\int_{y=0}^{\infty}
e^{-sy}
\left\{
\int_{x=0}^{\infty}
\E[e^{-s\hat{V}_{\mg}(y)} \mid \hat{V}_{\mg}(0) = x]
\,\dd F_{D_{n-1}}(x)
\right\}
\dd F_G(y),
\end{align*}
which together with (\ref{eq:hatV-transient}) implies (\ref{eq:A_peak-by-D}).
\end{proof}

\begin{theorem}
\label{theorem:A-by-D}

The LST $f_A^*(\cdot)$ of the stationary AoI is given in terms of the
LST $f_D^*(\cdot)$ of the stationary system delay by
\begin{equation}
f_A^*(s) 
= 
f_H^*(s/\mu)
\int_{y=0}^{\infty}
\frac{1 - e^{-sy}}{\E[G]s}
\cdot
\mathcal{L}_{\omega}^{-1}
\left[
\frac{1}{\phi(s) - \omega}
\left(
\frac{s}{\psi(\omega)}
\cdot
f_D^*(\psi(\omega)) - f_D^*(s)
\right)
\right]\! (y)
\; \dd F_G(y).
\label{eq:A-by-D}
\end{equation}
\end{theorem}
\begin{proof}
Similarly to the proof of Lemma \ref{lemma:A_peak-by-D}, 
we observe from $D_n = (V(\alpha_n-) + H_n)/\mu$ that its
distribution equals that of the sum of $H/\mu$ and
$\hat{V}_{\mg}(G_n)$ conditioned on $\hat{V}_{\mg}(0) = D_{n-1}$.
We thus have for the stationary system
delay $D$ (cf.\ (\ref{eq:hatV-transient})),
\begin{equation}
f_D^*(s) 
= 
f_H^*(s/\mu)
\int_{y=0}^{\infty}
\mathcal{L}_{\omega}^{-1}
\left[
\frac{1}{\phi(s) - \omega}
\left(
\frac{s}{\psi(\omega)}
\cdot
f_D^*(\psi(\omega)) - f_D^*(s)
\right)
\right]\! (y)
\; \dd F_G(y).
\label{eq:D-recursion}
\end{equation}
Therefore, we obtain (\ref{eq:A-by-D}) by combining 
(\ref{eq:f_A*-formula}), (\ref{eq:A_peak-by-D}), and
(\ref{eq:D-recursion}).
\end{proof}

\michel{Theorem \ref{theorem:A-by-D} give rise to an expression for the mean AoI, that will be relied on in Section \ref{sec:PH}, provided in the following lemma.}

\yoshiaki{
\begin{corollary}
The mean AoI $\E[A]$ \michel{can be expressed} in terms of the
LST $f_D^*(\cdot)$ of the stationary system delay \michel{via}
\begin{align}
\E[A] 
&=
\frac{\E[H]}{\mu}
+
\frac{\E[G^2]}{2\E[G]}
+
\int_{y=0}^{\infty}
q(y) \cdot \frac{y}{\E[G]} \, \dd F_G(y),
\label{eq:EA-by-q(y)-0}
\end{align}
where $q(t)$ is given by
\begin{align}
q(t) :=& \:\E\bigl[\hat{V}_{\mg}(t) \mid \hat{V}_{\mg}(0) \mbox{ is distributed as $D$}\bigr]
\label{eq:q-as-mean}
\\
=&\:
\mathcal{L}_{\omega}^{-1}\left[
-\frac{1-\rho_{\bg}}{\omega^2}
+
\frac{1}{\omega}
\left(
\frac{f_D^*(\psi(\omega))}{\psi(\omega)}
+\E[D]
\right)
\right].
\label{eq:q-t-inv}
\end{align}

\end{corollary}
\begin{proof}

We rewrite (\ref{eq:A-by-D}) as
\begin{equation}
f_A^*(s) 
= 
f_H^*(s/\mu)
\int_{y=0}^{\infty} \frac{1 - e^{-sy}}{\E[G]s} \cdot p(s,y)
\; \dd F_G(y),
\label{eq:f_A-by-p-st}
\end{equation}
where $p(s,t)$ is defined as
\[
p(s,t) 
:= 
\mathcal{L}_{\omega}^{-1}
\left[
\frac{1}{\phi(s) - \omega}
\left(
\frac{s}{\psi(\omega)}
\cdot
f_D^*(\psi(\omega)) - f_D^*(s)
\right)
\right]\! (t)
\]
Note here that $p^*(s, t)$ represents the LST of the transient virtual waiting time
of the M/GI/1 queue \michel{{\it with only the background traffic}}, whose initial
value follows the same distribution as the stationary system delay $D$
(cf.\ (\ref{eq:hatV-transient})):
\begin{equation}
p(s,t)
=
\E\bigl[e^{-s\hat{V}_{\mg}(t)} \mid \hat{V}_{\mg}(0) \mbox{ is
distributed as $D$}\bigr].
\label{p(s,t)-by-Z}
\end{equation}
Noting that
\[
q(t) = \lim_{s \to 0+} (-1)\cdot \frac{\dd p(s,t)}{\dd s},
\]
we have from (\ref{eq:f_A-by-p-st}), 
\begin{align*}
\E[A] &= \lim_{s \to 0+} (-1)\cdot \frac{\dd f_A^*(s)}{\dd s}
\\
&=
\frac{\E[H]}{\mu}
+
\frac{1}{\E[G]}
\int_{y=0}^{\infty}
\left( \frac{y}{2} + q(y)\right)y \, \dd F_G(y).
\end{align*}
Therefore, we obtain (\ref{eq:EA-by-q(y)-0}) and (\ref{eq:q-as-mean}).
In addition, (\ref{eq:q-t-inv}) is shown as follows:
\begin{align*}
\int_0^{\infty} q(y) e^{-\omega y} \dd y
&=
(-1)\cdot 
\int_0^{\infty} \lim_{s \to 0+}\frac{\partial p(s,y)}{\partial s} \cdot e^{-\omega y} \dd y
\\
&=
(-1)\cdot \lim_{s \to 0+} \frac{\partial p^*(s,\omega)}{\partial s}
=
-\frac{1-\rho_{\bg}}{\omega^2}
+
\frac{1}{\omega}
\left(
\frac{f_D^*(\psi(\omega))}{\psi(\omega)}
+\E[D]
\right).
\qedhere
\end{align*}
\end{proof}
}

\yoshiaki{
\begin{remark}
\label{remark:A-by-D}

\michel{The representation for $f_A^*(s)$ has an appealing interpretation, in terms of a decomposition. Indeed, we can} rewrite expression (\ref{eq:A-by-D}) as
\begin{align*}
f_A^*(s) 
&= 
f_H^*(s/\mu)
\int_{y=0}^{\infty}
\frac{1 - e^{-sy}}{sy}
\cdot
\mathcal{L}_{\omega}^{-1}
\left[
\frac{1}{\phi(s) - \omega}
\left(
\frac{s}{\psi(\omega)}
\cdot
f_D^*(\psi(\omega)) - f_D^*(s)
\right)
\right]\! (y)
\; \frac{y\,\dd F_G(y)}{\E[G]}
\\
&=
f_H^*(s/\mu)
\int_{y=0}^{\infty}
\E[e^{-s \mathrm{Unif}(0,y)}]
\cdot
\E[e^{-s Z(y; D)}] \; \dd F_{\check{G}}(y),
\end{align*}
where $\mathrm{Unif}(0,y)$ denotes \michel{a uniformly distributed} random variable on
$[0,y]$, \[Z(y; D) := [\hat{V}_{\mg}(y) \mid \hat{V}_{\mg}(0) = D]\]
denotes the transient virtual waiting time of the M/GI/1 queue with
initial workload $D$, and $\check{G}$ denotes the length-biased version
of the inter-generation time $G$.
This alternative expression shows that, \michel{conditional on} a value of
$\check{G}$, the stationary AoI $A$ \michel{can be} decomposed into the sum of
three (conditionally) independent random variables. \michel{Indeed, we have the distributional equality}
\[
[A \,|\, \check{G}=y] =_{\st} \frac{H}{\mu} + \mathrm{Unif}(0,y) + Z(y;D),
\]
\michel{where `$=_{\st}$' denotes that both sides of the equation have the same distribution,} and 
where $\check{G}$ and the terms on the right-hand side of this
equation are mutually independent.
In particular, we have \michel{for the mean that}
\begin{align*}
\E[A] &= 
\frac{\E[H]}{\mu} + \E[\mathrm{Unif}(0,\check{G})] + \E[Z(\check{G};D)]
\\
&= 
\frac{\E[H]}{\mu} + \frac{\E[G^2]}{2\E[G]} + \E[Z(\check{G};D)].
\end{align*}
Therefore, \michel{the} analysis of the \michel{mean} stationary AoI $A$ in this model is
essentially reduced to that of \michel{${\mathbb E}[Z(y; D)]$}, i.e., \michel{the expected value of the transient virtual
waiting time} of the M/GI/1 queue after the generation of a packet from
the tagged stream.
\end{remark}
}

The LST of the stationary AoI is thus given in terms of that of the
stationary system delay $D$. The characterization of $D$ is known in
the literature \cite{Hooke1972,Ott1984}, which takes a form of an
intuitively appealing decomposition formula:
\begin{lemma}[\!\!{\cite{Hooke1972,Ott1984}}]
\label{lemma:D-decomp}
Consider an GI/GI/1 queue which has inter-arrival times following the CDF
$F_G(\cdot)$ and service times with LST 
$f_{H_{\star}}^*(s) := f_H^*( (s+\lambda_{\bg}-\lambda_{\bg}f_{B_{\bg}}^*(s))/\mu )$ ($s \geq 0$).
Let $W_{\star}$ denote the stationary waiting time in this GI/GI/1 queue:
\begin{equation}
f_{W_{\star}}^*(s)
=
\exp\left[
-\sum_{k=1}^{\infty} \frac{1}{k}
\int_0^{\infty}
(1-e^{-sx})\, \dd F_{U_k}(x)
\right],
\quad
s \geq 0,
\label{eq:f_W_star}
\end{equation}
where $U_k$ denotes the sum of $k$ i.i.d.\ copies of the difference
$H_{\star} - G$ of the inter-arrival and service times, i.e.,
\begin{align}
F_{U_1}(x) 
&= 
\int_0^{\infty} \Pr( \michel{H_{\star}\leq x+y})\, \dd F_G(y),
\;\;
x \in (-\infty, \infty),
\label{eq:U_1-def}
\\
F_{U_k}(x)
&=
\int_{-\infty}^{\infty}
F_{U_{k-1}}(x-y)\,
\dd F_{U_1}(y),
\;\;
x \in (-\infty, \infty),
\,
k = 2,3,\ldots.
\label{eq:U_k-def}
\end{align}
We then have for the original model with GI/GI and M/GI input
streams, the LST of the system delay $D$ experienced by the GI/GI
stream is given by
\[
f_D^*(s) 
= 
\frac{(1-\rho_{\bg})s}
{s-\lambda_{\bg}+\lambda_{\bg}f_{H_{\bg}}^*(s/\mu)}
\cdot
f_{W_\star}^*(s - \lambda_{\bg}+\lambda_{\bg}f_{H_{\bg}}^*(s/\mu))
\cdot f_H^*(s/\mu).
\]
\end{lemma}
\begin{remark}
The expression (\ref{eq:f_W_star}) for the stationary waiting time
$W_{\star}$ in the GI/GI/1 queue can be found in \cite[P.\ 280]{Cohen1982}.
Its mean $\E[W_{\star}]$ is also given explicitly by \cite[P.\ 232]{Asmussen2003}:
\begin{equation}
\E[W_{\star}] = \sum_{k=1}^{\infty}\frac{1}{k} \int_0^{\infty} x \,\dd
F_{U_k}(x).
\label{eq:E_W_star}
\end{equation}
\end{remark}
\begin{proof}
Lemma \ref{lemma:D-decomp} readily follows from
\cite{Hooke1972,Ott1984}. Note that although these papers show 
the decomposition formula for only the time-depedent
and stationary virtual waiting time $\hat{V}(t)$ and $V$, their
discussion clearly carries over to the virtual waiting time just before
arrivals of the GI/GI stream.
\end{proof}

\begin{remark}
\label{remark:AoI-priorwork}
\michel{
In the prior work \cite{Inoue2019} a {\em single-source} model was considered, thus facilitating and analysis that intensively uses the Lindley recursion; see e.g.\ \cite[Ch.\ I, Eq.\ (5.11)]{Asmussen2003}. In the setup of the present paper a significant complication arose, in that we could not use this Lindley-type equation for deriving the peak AoI, since we needed to account for the packets from the background stream between two successive arrivals from the tagged stream. The analysis presented above resolves this issue.}
\end{remark}

Theorem \ref{theorem:A-by-D} and Lemma \ref{lemma:D-decomp} completely
determine the AoI distribution. In particular, \michel{they allow us to evaluate} the
value of $f_A^*(s)$ and its derivatives numerically, which enables us
to compute the moments and the CDF $F_A(\cdot)$ of the stationary 
AoI; we refer to e.g.\ \cite{Abate1992} for Laplace inversion methods to numerically compute the corresponding CDF.
\michel{This being said, these formulas do not provide explicit insight 
into the way the model parameters affect the AoI performance. 
In the next section, we explore a more explicit characterization utilizing
stochastic ordering techniques, with a focus on the mean stationary AoI.}

\section{Stochastic Orderings and Bounds \michel{on} the AoI}\label{sec:order}

\yoshiaki{
In this section, we derive stochastic comparison results and
closed-form bounds for the AoI \michel{by applying} stochastic ordering
techniques. 
\michel{The results of Section \ref{sec:exact} describe the full distribution of the AoI, but it requires numerical inversion to obtain the density or the cumulative distribution function. An important asset of the results of the present section is that the bounds provide {\it explicit} insight into the impact of the various model parameters on the mean AoI. 
This is even more relevant in light of the fact that, as backed by the experiments presented in Section \ref{sec:num}, the derived {\it upper bound} is typically remarkably close to the true value of the mean AoI.}

\yoshiaki{\michel{We start the section}
by presenting \michel{its} main results. \michel{The first result, Theorem~\ref{theorem:A-bound}, relies on the concept of {\it stochastic ordering}.}
\michel{With $Y_0$ and $Y_1$ denoting} non-negative random variables,
$Y_0$ is said to be smaller than $Y_1$ in \emph{the usual stochastic
order} (denoted by $Y_0 \leq_{\st} Y_1$) when
\begin{equation}
\overline{F}_{Y_0}(x) \leq \overline{F}_{Y_1}(x),
\quad
\forall x \geq 0.
\label{eq:u-st}
\end{equation}
It is known that $Y_0 \leq_{\st} Y_1$ if and only if
$\E[h(Y_0)] \leq \E[h(Y_1)]$ holds for all non-decreasing functions
$h(\cdot)$ such that the expectations exist \cite{Shaked2007}.
\michel{Theorem~\ref{theorem:A-bound}, to be proven in in Section \ref{ssec:stoch-comp}, provides (surprisingly simple) distributional
bounds on the stationary AoI.} 
\begin{theorem}
\label{theorem:A-bound}
The stationary AoI $A$ is bounded by the sum of independent random variables as
\begin{equation}
D^- + \widetilde{G} \leq_{\st} A \leq_{\st} D + \widetilde{G},
\label{eq:A-dist-bounds}
\end{equation}
where $D$ denotes the stationary system delay as before, $D^-$ denotes
the system delay in the ordinary M/GI/1 queue with
only the background traffic existing, and $\widetilde{G}$ denotes the 
generic random variable for residual \michel{inter-generation} times:
\begin{align*}
f_{D^-}^*(s) &= \frac{(1-\rho_{\bg})s}
{s-\lambda_{\bg}+\lambda_{\bg}f_{H_{\bg}}^*(s/\mu)},
\\
f_{\widetilde{G}}(x) &= \frac{\overline{F}_G(x)}{\E[G]},
\quad
f_{\widetilde{G}}^*(s) = \frac{1 - f_G^*(s)}{s\E[G]}.
\end{align*}
\end{theorem}

To understand the intuition behind the bounds \michel{in} (\ref{eq:A-dist-bounds}), \michel{we}
refer to (\ref{eq:A_peak-DG}) and (\ref{eq:f_A*-formula}):
from these equations \michel{it is seen} that if the inter-generation time
$G_n$ and the system delay $D_n$ \textit{were} independent, the
right-hand side of (\ref{eq:f_A*-formula}) reduces to
\[
\frac{f_D^*(s) - f_D^*(s)f_G^*(s)}{s\E[G]} 
=
f_D^*(s) \cdot \frac{1 -f_G^*(s)}{s\E[G]} 
=
f_{D+\widetilde{G}}^*(s).
\]
This means that the construction of the lower bound $A \geq_{\st} D^- +
\widetilde{G}$ in (\ref{eq:A-dist-bounds}) is relatively straightforward,
because it \michel{amounts} to ignoring the traffic load \michel{caused} by the
tagged stream, which decreases the magnitude of the system delay $D_n$,
while making $G_n$ and $D_n$ independent.

On the other hand, the upper bound $A \leq_{\st} D + \widetilde{G}$ is
far from trivial, because it states that just ignoring the dependence of
$G_n$ and $D_n$, with their magnitude retained, yields an upper bound \michel{on} the stationary AoI $A$.
We note that a similar inequality \michel{was} established in \cite{Inoue2019} 
for the mean AoI $\E[A]$ in a single-source GI/GI/1
queue, which is based on the well-known expression \cite{Kaul2012}:
\[
\E[A] = \frac{\E[G_nD_n]}{\E[G]} + \frac{\E[G^2]}{2\E[G]} 
=
\E[D] + \E[\widetilde{G}] + \frac{\Cov[G_n,D_n]}{\E[G]}.
\]
In \cite{Inoue2019}, it is shown that $G_n$ and $D_n$ are negatively
correlated (i.e., $\Cov[G_n,D_n] \leq 0$), so that $\E[A] \leq \E[D]
+\E[\widetilde{G}]$ holds for the single-source GI/GI/1 queue.
\michel{There are two crucial differences with our result, though.
In the first place, the proof of the result in \cite{Inoue2019}} heavily relies on the Lindley recursion for the
GI/GI/1 queue, which is not applicable to our model with the background
stream. Furthermore, the upper bound (\ref{eq:A-dist-bounds}) is
stronger than \michel{the one that appeared in \cite{Inoue2019}},} as it proves \michel{a {\it distributional} bound}: it provides not only \michel{an upper bound on} the mean AoI $\E[A]$ itself, but also \michel{an upper bound on} the non-linear AoI
penalty $\E[g(A)]$ for any (non-decreasing) AoI penalty function $g$.

\medskip

\michel{We now move to the second main result of this section: 
utilizing the stochastic ordering result of Theorem \ref{theorem:A-bound}, Theorem~\ref{theorem:A-bound-mean} provides} a
closed-form bound on the mean AoI by focusing on the case where
\textit{\michel{inter-generation} times have relatively small
variability}. Note that choosing less variable \michel{inter-generation} times is reasonable
because it leads to the reduction in the AoI as known in the literature \cite{Kaul2012,Inoue2019}.
Specifically, we derive a simpler bound on the mean AoI assuming
that the \michel{inter-generation} time distribution is `new better than used in
expectation' (NBUE), i.e.,
\[
\E[G-y \,|\, G > y] \leq \E[G],
\quad
\forall y \geq 0.
\]
Intuitively, for NBUE generation intervals, the expected residual
generation time does not increase as time passes.
The class of NBUE distributions includes gamma distributions with
shape parameter not smaller than $1$ (including the exponential
distribution) and the deterministic distribution,
which are typical choices when modeling \michel{inter-generation} times in
monitoring systems.

\begin{theorem}\label{theorem:A-bound-mean}
If the \michel{inter-generation} time distribution $F_G(\cdot)$ is NBUE, 
we have
\begin{align*}
\frac{\lambda_{\bg}\E[H_{\bg}^2]}{2(1-\rho_{\bg})\mu^2}
+ 
\frac{\E[H]}{\mu}
+
\frac{\E[G^2]}{2\E[G]}
\leq \E[A] \leq 
\frac{\lambda\E[\yoshiaki{H^2}]+\lambda_{\bg}\E[H_{\bg}^2]}{2(1-\rho-\rho_{\bg})\mu^2}
+ 
\frac{\E[H]}{\mu}
+ 
\frac{\E[G^2]}{2\E[G]}.
\end{align*}

\end{theorem}

As mentioned above, the \textit{upper} bound in Theorem \ref{theorem:A-bound-mean}
is particularly relevant as it is typically remarkably close to
the true value of the mean AoI, \michel{a claim corroborated by the experiments presented in Section \ref{sec:num}.}

\medskip

The rest of this section is organized as follows.
In Section \ref{ssec:stoch_orders}, we briefly review the
definitions of \michel{the types of stochastic orderings that we use}. 
In Section \ref{ssec:stoch-comp}, we prove several stochastic ordering
properties of the AoI, which leads to its distributional bounds,
\yoshiaki{\michel{most notably} Theorem \ref{theorem:A-bound} presented above.}
We then derive closed-form bounds on the mean AoI, \yoshiaki{including
Theorem \ref{theorem:A-bound-mean}}, in Section \ref{ssec:bounds}.
}

\subsection{Stochastic \michel{Ordering}}
\label{ssec:stoch_orders}

\yoshiaki{We briefly review} the definitions of some stochastic orders (see
\cite{Shaked2007} for more detailed explanations).
Let $Y_0$ and $Y_1$ denote two non-negative random valiables.
$Y_0$ is said to be smaller than $Y_1$ in \emph{the
increasing convex order} (denoted by $Y_0 \leq_{\icx} Y_1$) when
\begin{equation}
\int_x^{\infty} \overline{F}_{Y_0}(y)\,\dd y 
\leq 
\int_x^{\infty} \overline{F}_{Y_1}(y)\,\dd y,
\quad
\forall x \geq 0.
\label{eq:icx-def-int}
\end{equation}
As the name suggests, $Y_0 \leq_{\icx} Y_1$ if and only if 
$\E[h(Y_0)] \leq \E[h(Y_1)]$ holds for all non-decreasing convex
functions $h(\cdot)$ such that the expectations exist.
Similarly, $Y_0$ is said to be smaller than $Y_1$ in \emph{the
increasing concave order} (denoted by $Y_0 \leq_{\icv} Y_1$) when
\[
\int_0^x \overline{F}_{Y_0}(y)\,\dd y 
\leq 
\int_0^x \overline{F}_{Y_1}(y)\,\dd y,
\quad
\forall x \geq 0,
\]
which is equivalent to that $\E[h(Y_0)] \leq \E[h(Y_1)]$ holds for all
non-decreasing concave functions $h(\cdot)$ such that the expectations exist.

Furthermore, $Y_0$ is said to be smaller than $Y_1$ in \emph{the
convex order} (denoted by $Y_0 \leq_{\cx} Y_1$) when 
\[
Y_0 \leq_{\icx} Y_1\quad \mbox{and}\quad \E[Y_0] = \E[Y_1].
\]
The convex ordering $Y_0 \leq_{\cx} Y_1$ holds if and only if $\E[h(Y_0)] \leq \E[h(Y_1)]$
for all convex functions $h(\cdot)$ such that the expectations exist.
Because $\E[Y_i] = \int_0^{\infty} \overline{F}_{Y_i}(y)\,\dd y$, it is
readily seen that $Y_0 \leq_{\cx} Y_1$ is also equivalent to
$Y_0 \geq_{\icv} Y_1$ and $\E[Y_0] = \E[Y_1]$; note that the
inequalities have opposite directions here.

While the usual stochastic order $\leq_{\st}$ \yoshiaki{(defined 
in (\ref{eq:u-st}))} compares the magnitude of random variables, the
convex order $\leq_{\cx}$ compares random
variables in terms of the variability given the same mean.
The increasing convex and increasing concave orders compare both the
magnitude and variability. Specifically, $Y_0 \leq_{\icx} Y_1$ implies
that $Y_0$ is smaller and less variable than $Y_1$, while 
$Y_0 \leq_{\icv} Y_1$ implies that $Y_0$ is smaller but more variable
than $Y_1$.

Notice however that $\leq_{\st}$ is stronger than $\leq_{\icx}$ in the sense that 
$Y_0 \leq_{\st} Y_1 \Rightarrow Y_0 \leq_{\icx} Y_1$ (also we have
$Y_0 \leq_{\st} Y_1 \Rightarrow Y_0 \leq_{\icv} Y_1$).
On the other hand, if $Y_0 \leq_{\st} Y_1$ and $Y_0 \leq_{\cx} Y_1$,
then $Y_0 =_{\st} Y_1$ \cite[Th.\ 1.A.8]{Shaked2007}, where $=_{\st}$
denotes the equality in distribution \yoshiaki{(cf.\ Remark \ref{remark:A-by-D})}. 
Therefore, the usual stochastic
order and the convex order compare random variables from essentially
different viewpoints.

\subsection{Distributional Bounds for the AoI}
\label{ssec:stoch-comp}

\michel{In this subsection we prove Theorem \ref{theorem:A-bound}.}
As mentioned in the proof of Lemma \ref{lemma:A_peak-by-D}, 
the virtual waiting time process $(\hat{V}(t))_{t \in [\alpha_{n-1},\alpha_n)}$ during
an \michel{inter-generation} time behaves in the same way as the ordinary M/GI/1
queue with the arrival rate $\lambda_{\bg}$ and the service time distribution
$F_{\hat{H}_{\bg}}(\cdot)$, where $\hat{H}_{\bg} =_{\st} H/\mu$.
Recall that the virtual waiting time and the system delay are related as 
$D_n = \hat{V}(\alpha_{n})$.  We then define $(Z(t; D_n))_{t \geq 0}$ as the
virtual waiting time process of the M/GI/1 queue with only the
background traffic, whose initial value equals to the system
delay $D_n$:
\begin{equation}
\overline{F}_{Z(t; D_n)}(x)
=
\int_0^{\infty} 
\Pr(\hat{V}_{\mg}(t) > x \mid \hat{V}_{\mg}(0) = x_0)
\,\dd F_{D_n}(x_0),
\quad
x \geq 0.
\label{eq:F_Z-tx-orig-def}
\end{equation}
To make the notation simpler, we define
\begin{equation}
\overline{F}_{Z}(t, x; D_n)
=
\overline{F}_{Z(t; D_n)}(x), 
\quad
x \geq 0.
\label{eq:F_Z-tx-def}
\end{equation}
Let $W_n$ denote the waiting time of the $n$th packet 
(see (\ref{eq:D_n-by-V}) and (\ref{eq:hatV-def})): 
\[
W_n = \frac{V(\alpha_n-)}{\mu} = \hat{V}(\alpha_n-).
\]
By definition, we have 
\begin{align}
\overline{F}_{W_n}(x) 
&= 
\int_0^{\infty}
\overline{F}_{Z}(t, x; D_{n-1})\,
\dd F_G(t),
\quad
n = 1,2,\ldots.
\label{eq:D_n-recursion-by-F}
\end{align}
Note that (\ref{eq:D_n-recursion-by-F}) along with 
\begin{equation}
D_n = W_n + \frac{H_n}{\mu},
\label{eq:D_n-by-W_n}
\end{equation}
defines a recursive relation for $D_n$ ($n=0,1,\ldots$).

Imposing the following assumption on the initial virtual waiting time
$\hat{V}(0)$ largely simplifies the derivation of stochastic
comparison results. It is by no means restrictive, as such an assumption on $\hat{V}(0)$ does not affect any of the
results on stationary distributions.
\begin{assumption}
\label{asm:V(0)-mg1}
The initial virtual waiting time $\hat{V}(0)$ is distributed according
to the stationary virtual waiting time distribution in the M/GI/1
queue with only the background traffic existing, i.e.,
\[
\E[e^{-s\hat{V}(0)}] = \frac{(1-\rho_{\bg})s} {s-\lambda_{\bg}+\lambda_{\bg}f_{H_{\bg}}^*(s/\mu)},
\quad
s > 0.
\]
\end{assumption}
\begin{remark}
Assumption \ref{asm:V(0)-mg1} represents the situation that the system
had been running only with the background traffic until the first
generation time $\alpha_0$ of the packet by the sensor.
\end{remark}

Under this assumption, the following monotonicity property is known in the literature \cite{Ott1984}.
\begin{lemma}[\!\!{\cite[Th.\ 1.1]{Ott1984}}]
\label{lemma:Z-decreasing}

If Assumption \ref{asm:V(0)-mg1} is satisfied, then $Z(t; D_n)$ is
stochastically decreasing in $t$ for each $n$, i.e.,
\[
\overline{F}_Z(t_1,x; D_n) \geq \overline{F}_Z(t_2,x; D_n),
\quad
0 \leq t_1 \leq t_2,
\;\;
x \geq 0,
\;\;
n = 1,2,\ldots.
\]
\end{lemma}

Informally, the virtual waiting time process $\hat{V}(t)$ during
the time-interval $t \in [\alpha_{n-1}, \alpha_n)$
starts from the system delay $D_{n-1} = \hat{V}(\alpha_{n-1})$ and it tends
toward the stationary distribution of the M/GI/1 queue.
Lemma \ref{lemma:Z-decreasing} shows that this transition is in fact
decreasing in the usual stochastic order.

Recall that the exact peak AoI distribution is characterized by the
rather complicated expression in (\ref{eq:A_peak-by-D}).
The following lemma shows that \textit{ignoring the dependence
between $D_n$ and $G_n$} in (\ref{eq:A_peak-DG}) yields a simple upper
bound of the peak AoI in the convex order.
\begin{lemma}
\label{lemma:A_peak-cx-bound}

Under Assumption \ref{asm:V(0)-mg1}, the $n$th peak AoI $A_{\peak,n}$
is bounded above in the convex
order as
\[
A_{\peak,n} \leq_{\cx} D_n + G,
\]
where $G$ denotes a generic random variable for \michel{inter-generation} times
and it is independent of the system delay $D_n$.
\end{lemma}
\begin{proof}

We consider the joint distribution of $W_n$ and the \michel{inter-generation}
time $G_n$ between $(n-1)$st and $n$th packets in the following form: 
\begin{align*}
\lefteqn{
\Pr(W_n > x, G_n \leq y)
}\qquad
\\
&=
\int_{x_0=0}^{\infty}
\Pr(\hat{V}_{\mg}(G_n) > x, G_n \leq y \mid \hat{V}_{\mg}(0) = x_0)
\,\dd F_{D_{n-1}}(x_0)
\\
&=
\int_{t=0}^{\infty}
\id{t \leq y}
\int_{x_0=0}^{\infty}
\Pr(\hat{V}_{\mg}(t) > x \mid \hat{V}_{\mg}(0) = x_0)
\,\dd F_{D_{n-1}}(x_0)
\,\dd F_G(t)
\\
&=
\int_{t=0}^{\infty}
\id{t \leq y}
\overline{F}_Z(t,x; D_{n-1})
\,\dd F_G(t),
\end{align*}
where $\overline{F}_Z(t,x; D_{n-1})$ is defined as in (\ref{eq:F_Z-tx-orig-def}) and (\ref{eq:F_Z-tx-def}).
We then have
\begin{align*}
\Pr(W_n > x, G_n \leq y)
&=
\E[\overline{F}_Z(G,x; D_{n-1})\id{G \leq y}]
\geq\,
\E[\overline{F}_Z(G,x; D_{n-1})]\E[\id{G \leq y}],
\end{align*}
where the inequality follows from the fact that 
for any random variable $Y$ and non-increasing functions
$g_1(\cdot)$ and $g_2(\cdot)$, we have $\E[g_1(Y)g_2(Y)] \geq
\E[g_1(Y)]\E[g_2(Y)]$ (see \cite{Gurland1967,Behboodian1994}).

Because (\ref{eq:D_n-recursion-by-F}) implies 
$\Pr(W_n > x) = \E[\overline{F}_Z(x,G; D_{n-1})]$,
we further obtain
\[
\Pr(W_n > x, G_n \leq y)
\geq
\Pr(W_n > x)\Pr(G_n \leq y),
\quad
x \geq 0,
\;\;
y \geq 0.
\]
Therefore, $W_n$ and $G_n$ are 
\textit{negatively quadrant dependent} \cite[Lemma 1]{Lehmann1966}.
It then follows from \cite[Th.\ 1 (ii) and Lemma 3]{Lehmann1966}
that for any non-decreasing functions
$h_1(\cdot)$ and $h_2(\cdot)$,
\[
\Cov[h_1(W_n), h_2(G_n)] \leq 0.
\]
We then have from \cite[Th.\ 3.A.39]{Shaked2007},
\begin{equation}
W_n + G_n \leq_{\cx} W_n + G,
\label{eq:V+G-cx-order}
\end{equation}
where $G$ on the right-hand side is taken to be independent of
$W_n$. Because the convex order is preserved 
under the addition of an independent random variable \cite[Eq.\ (3.A.46)]{Shaked2007},
we obtain from (\ref{eq:A_peak-DG}), (\ref{eq:D_n-by-W_n}), and 
(\ref{eq:V+G-cx-order}), 
\begin{align*}
A_{\peak,n+1} &= W_n + G_n + \frac{H_n}{\mu}
\\
&\leq_{\cx}
W_n + G + \frac{H_n}{\mu}
=
D_n + G.
\qedhere
\end{align*}
\end{proof}
As an application of Lemma \ref{lemma:A_peak-cx-bound}, a similar bound for the
stationary version follows immediately, as stated in the next corollary. 
\begin{corollary}
\label{cor:A_peak-cx-bound}

The stationary peak AoI $A_{\peak}$ is bounded above in the convex
order as
\[
A_{\peak} \leq_{\cx} D + G,
\]
where the stationary system delay $D$ and the \michel{inter-generation} time
$G$ are taken to be independent.
\end{corollary}
\begin{remark}
Since Corollary \ref{cor:A_peak-cx-bound} refers only to the stationary
distribution, Assumption \ref{asm:V(0)-mg1} is irrelevant here.
\end{remark}

\yoshiaki{
\michel{By} Corollary \ref{cor:A_peak-cx-bound} and an additional \michel{argumentation} for
a lower bound, we \michel{can} obtain Theorem \ref{theorem:A-bound} as follows.
}

\medskip

\noindent
\yoshiaki{
\textit{Proof of Theorem \ref{theorem:A-bound}. }}
We first derive the upper bound.
From (\ref{eq:A_peak-DG}), (\ref{eq:icx-def-int}), and Corollary \ref{cor:A_peak-cx-bound}, we have
\begin{align*}
\int_0^x \overline{F}_{A_{\peak}}(y)\,\dd y 
&=
\E[A_{\peak}] - \int_x^{\infty} \overline{F}_{A_{\peak}}(y)\,\dd y 
\\
&=
\E[D]+\E[G] - \int_x^{\infty} \overline{F}_{A_{\peak}}(y)\,\dd y
\\
&\geq
\E[D]+\E[G] - \int_x^{\infty} \overline{F}_{D+G}(y)\,\dd y
=
\int_0^x \overline{F}_{D+G}(y)\,\dd y,
\quad
\forall x \geq 0.
\end{align*}
It then follows from (\ref{eq:F_A-formula}) that
\begin{align*}
F_A(x) 
&\geq 
\frac{1}{\E[G]}
\left\{
\int_0^x \overline{F}_{D+G}(y)\,\dd y
-
\int_0^x \overline{F}_{D}(y)\,\dd y
\right\}
\\
&=
\frac{1}{\E[G]}
\left\{
\int_{y=0}^x 
\left(
\overline{F}_D(y)
+
\int_{u=0}^y \overline{F}_G(y-u) \,\dd F_D(u)
\right) \dd y
-
\int_0^x \overline{F}_{D}(y)\,\dd y
\right\}
\\
&=
\int_{y=0}^x 
\left(
\int_{u=0}^y 
\frac{\overline{F}_G(y-u)}{\E[G]} \,\dd F_D(u)
\right)
\dd y
=
F_{D+\widetilde{G}}(x),
\quad
\forall x \geq 0,
\end{align*}
which proves the upper bound $A \leq_{\st} D+\widetilde{G}$.

We then consider the lower bound. Consider a modified workload process
$(V^-(t))_{t \geq 0}$
where the workload brought by the sensor after time $0$ is replaced by zero (cf.\ (\ref{eq:V(t)-def})):
\[
V^-(t) = V(0) + X_{\bg}^-(t) - \mu t + \mu \int_{u \in [0,t)}
\id{V^-(u) = 0} \dd u,
\quad
t \geq 0.
\]
It is obvious that we have $V^-(t) \leq V(t)$ (for $t \geq 0$).
We define $D_n^-$ (for $n = 0,1,\ldots$) by
\begin{equation}
D_n^- = \frac{V^-(\alpha_n)+H_n}{\mu},
\label{eq:D_n^-}
\end{equation}
which, together with (\ref{eq:D_n-by-V}), implies $D_n^- \leq D_n$ for each sample
path. Let us then define $A^-(t)$ as
\[
A^-(t) := t - \max\{\alpha_n;\, \alpha_n + D_n^- \leq t\},
\]
which clearly satisfies
\begin{equation}
A^-(t) \leq A(t),
\label{eq:A(t)-lower}
\end{equation}
for each sample path (cf.\ (\ref{eq:A(t)-def})).
Furthermore, we define the corresponding peak AoI as (cf.\
(\ref{eq:A_peak-DG}))
\begin{equation}
A_{\peak,n}^- 
:= 
\lim_{t \to \alpha_n+D_n^- -}A^-(t)
=
D_n^- + G_n.
\label{eq:A_peak-}
\end{equation}

Now suppose that the initial workload $V(0)$ is distributed according
to the stationary distribution of $(V^-(t))_{t \geq 0}$.
Under this condition, $D_n^-$ defined in (\ref{eq:D_n^-}) has the same
distribution as $D^-$ and it is independent of $G_n$.
Therefore, we have from (\ref{eq:A_peak-}) that
\[
\overline{F}_{A_{\peak,n}^-}(y)
= 
\overline{F}_{D^-}(y)
+ \int_0^y \overline{F}_G(y-u) \,\dd F_{D^-}(u).
\]
Applying (\ref{eq:F_A-formula}) to the modified AoI process
$(A^-(t))_{t \geq 0}$, we obtain its stationary random variable
$A^-$ as
\[
F_{A^-}(x) 
=
\int_{y=0}^x
\left(
\int_{u=0}^y \frac{\overline{F}_G(y-u)}{\E[G]} \dd F_{D^-}(u)
\right)\dd y,
\]
i.e., we have $A^- =_{\st} D^- + \widetilde{G}$.
Because (\ref{eq:A(t)-lower}) still holds under this initial
condition, their time-averaged distributions satisfy 
\[
\frac{1}{T}\int_0^T \id{A^-(t) > x}\dd t 
\leq 
\frac{1}{T}\int_0^T \id{A(t) > x}\dd t,
\;\;
T > 0, x \geq 0.
\]
Letting $T \to \infty$ in this equation (cf.\ (\ref{eq:A-time-avg-dist})),
we thus obtain the lower bound in (\ref{eq:A-dist-bounds}).
\qed

\subsection{Closed-Form Bounds for the Mean AoI}
\label{ssec:bounds}

Theorem \ref{theorem:A-bound} enables us to obtain closed-form bounds
for the mean AoI.
\begin{corollary}
\label{corollary:EA-bound0}
The mean AoI $\E[A]$ is bounded by
\begin{align}
\E[A] 
&\geq 
\frac{\lambda_{\bg} \E[H_{\bg}^2]}{2(1-\rho_{\bg})\mu^2 }
+ \frac{\E[H]}{\mu}
+
\frac{\E[G^2]}{2\E[G]},
\label{eq:EA-lower}
\\
\E[A] 
&\leq 
\frac{\lambda_{\bg} \E[H_{\bg}^2]}{2(1-\rho_{\bg})\mu^2 }
+ \frac{\E[H]}{\mu}
+ \frac{\E[G^2]}{2\E[G]}
+ (1-\rho_{\bg})\E[W_{\star}],
\label{eq:EA-upper-W_star}
\end{align}
where $\E[W_{\star}]$ is given by (\ref{eq:E_W_star}).
\end{corollary}
\begin{proof}
From Lemma \ref{lemma:D-decomp} and Theorem \ref{theorem:A-bound}, we can readily verify 
(\ref{eq:EA-lower}) and (\ref{eq:EA-upper-W_star}).
\end{proof}

\begin{remark}
Recall that $\E[W_{\star}]$ denotes the mean waiting time in the GI/G/1 queue, where the LST of
the service time $H_{\star}$ is given by $f_{H_{\star}}^*(s) =
f_H^*((s+\lambda_{\bg}-\lambda_{\bg}f_{B_{\bg}}^*(s))/\mu)$.
We can then replace $\E[W_{\star}]$ in the upper bound (\ref{eq:EA-upper-W_star})
by an expression only with the first two moments of the model parameters
\cite{Daley1977}:
\[
\E[W_{\star}] 
\leq 
\frac{\Var[H_{\star}] +
\rho_{\star}(2-\rho_{\star})\Var[G]}{2(1-\rho_{\star})\E[G]},
\]
where
\[
\rho_{\star} = \frac{\rho}{1-\rho_{\bg}},
\quad
\E[H_{\star}] = \frac{\E[H]}{\mu(1-\rho_{\bg})},
\quad
\Var[H_{\star}] 
= 
\frac{\Var[H]}{\mu^2 (1-\rho_{\bg})^2}
+
\frac{\rho_{\bg}\E[H_{\bg}^2]}{\mu^2(1-\rho_{\bg})^3}.
\]

\end{remark}

\yoshiaki{
We are now in a position to prove Theorem \ref{theorem:A-bound-mean}.

\medskip

\noindent
\textit{Proof of Theorem \ref{theorem:A-bound-mean}. }}
Owing to Corollary \ref{corollary:EA-bound0}, it is sufficient to show
that
\begin{equation}
\frac{\lambda_{\bg}\E[H_{\bg}^2]}{2(1-\rho_{\bg})\mu^2}
+(1-\rho_{\bg})\E[W_{\star}] 
\leq
\frac{\lambda\E[H^2]+\lambda_{\bg}\E[H_{\bg}^2]}{2(1-\rho-\rho_{\bg})\mu^2}.
\label{eq:D-bound-NBUE}
\end{equation}
Recall that $W_{\star}$ denotes the stationary waiting time in the GI/GI/1 queue as in Lemma \ref{lemma:D-decomp}, which has the inter-arrival time distribution $F_G(\cdot)$ and the modified service time distribution $F_{H_{\star}}(\cdot)$ with LST $f_{H_{\star}}^*(s) =
f_H^*((s+\lambda_{\bg}-\lambda_{\bg}f_{B_{\bg}}^*(s))/\mu)$. From (\ref{eq:E_W_star}), $\E[W_{\star}]$ is given by
\[
\E[W_{\star}] = \sum_{k=1}^{\infty}\frac{1}{k} \cdot \E[\max(0,U_k)],
\]
where $U_k$ is defined as the sum of $k$ i.i.d.\ copies of the difference
$H_{\star} - G$ of the inter-arrival and service times (see
(\ref{eq:U_1-def}) and (\ref{eq:U_k-def})).

\yoshiaki{
It \michel{is} known that $G$ is NBUE if and only if \cite[Th.\ 3.A.55]{Shaked2007}
\begin{equation}
G \leq_{\cx} \mathrm{Exp}(\E[G]),
\label{eq:NBUE-cx}
\end{equation}
where $\mathrm{Exp(\E[G])}$ denotes an exponentially distributed
random variable with the same mean as~$G$.}
Because $\max(0, x)$ is a convex function of $x$, (\ref{eq:NBUE-cx}) implies that if $G$ is NBUE, then 
$\E[W_{\star}]$ is bounded above by the mean stationary waiting time in the M/G/1 queue with the same arrival rate $1/\E[G]$ and the same service time distribution, i.e.,
\begin{equation}
\E[W_{\star}] 
\leq 
\frac{\lambda\E[H_{\star}^2]}{2(1-\lambda \E[H_{\star}])}
=
\frac{1}{1-\rho_{\bg}}
\left\{
\frac{\lambda \E[H^2] + \lambda_{\bg} \E[H_{\bg}^2]}{2(1-\rho-\rho_{\bg})\mu^2}
-
\frac{\lambda_{\bg}\E[H_{\bg}^2]}{2(1-\rho_{\bg})\mu^2}
\right\},
\label{eq:W_star-NBUE-bound}
\end{equation}
where the last equality can be verified with a straightforward calculation using
\[
\lambda\E[H_{\star}] 
= \frac{\rho}{1-\rho_{\bg}},
\quad
\lambda\E[H_{\star}^2]
=
\frac{\rho\lambda_{\bg}\E[H_{\bg}^2]}{(1-\rho_{\bg})^3\mu^2}
+\frac{\lambda\E[H^2]}{(1-\rho_{\bg})^2\mu^2}.
\]
We then obtain (\ref{eq:D-bound-NBUE}) from (\ref{eq:W_star-NBUE-bound}).
\qed

\section{Phase-type \michel{inter-generation} time distribution}\label{sec:PH}

\michel{In this section, we develop a numerical algorithm to compute the mean AoI ${\mathbb E}[A]$, for the class of PH+M/GI+GI/1 systems, presented in Theorem \ref{thm_meanA}. It is recalled that one can approximate any distribution on the positive
half-line arbitrarily closely by a distribution from the phase-type
class described above \cite[Th.\ III.4.2]{Asmussen2003}. This motivates considering the queueing model PH+M/GI+GI/1, rather than the formally more general model GI+M/GI+GI/1. It is noted that the resulting recipe for evaluating  ${\mathbb E}[A]$ is particularly useful, as it can be used to assess the tightness of the closed-form bounds derived in Theorem \ref{theorem:A-bound-mean}; this will be done in Section \ref{sec:num}.}

We consider a special case where the \michel{inter-generation} times
$(G_n)_{n=1,2,\ldots}$ follow an $M$-state ($M=1,2,\ldots$) phase-type
distribution \cite[Section III.4]{Asmussen2003} with representation $(\bm{\gamma}, \bm{\Gamma})$:
\begin{align}
F_G(x) &= 1 - \bm{\gamma} \exp[\bm{\Gamma} x] \bm{e},
\quad
x \geq 0,
\nonumber
\\
f_G(x) &= \bm{\gamma} \exp[\bm{\Gamma} x] (-\bm{\Gamma})\bm{e},
\quad
x \geq 0,
\label{eq:f_G-PH}
\end{align}
where $\bm{e}$ denotes an $M\times 1$ vector with all elements equal
to one. The mean \michel{inter-generation} time is then given by
\begin{equation}
\E[G] = \bm{\gamma} (-\bm{\Gamma})^{-1} \bm{e}.
\label{eq:EG-PH}
\end{equation}
The phase-type distribution represents the absorption time of a Markov
chain with the transition rate matrix $\bm{\Gamma}$ started from an
initial state distributed according to the probability vector
$\bm{\gamma}$. By definition, we have 
$\bm{\gamma} \geq 0$ and $\bm{\gamma}\bm{e} = 1$. Also, $\bm{\Gamma}$
has negative diagonal elements and non-negative off-diagonal elements.
Because the Markov chain governed by $\bm{\Gamma}$ is absorbing,
$\bm{\Gamma}$ is a \textit{defective} transition rate matrix, i.e., its row
sums are non-positive (i.e., $\bm{\Gamma}\bm{e} \leq 0$) with at least
one of them being strictly negative.

We impose the common assumption that 
the transition rate matrix $\bm{\Gamma} +
(-\bm{\Gamma}\bm{e})\bm{\gamma}$ is irreducible, where 
$(-\bm{\Gamma}\bm{e})\bm{\gamma}$ corresponds to the event of
absorption and the transition to an initial state for the next \michel{inter-generation} time.
It is then ensured that  $\bm{\Gamma} + (-\bm{\Gamma}\bm{e})\bm{\gamma}$ 
has a unique stationary probability vector $\bm{\pi}$, which  solves
\[
\bm{\pi}
\{\bm{\Gamma} + (-\bm{\Gamma}\bm{e})\bm{\gamma}\}
= \bm{0},
\quad
\bm{\pi}\bm{e} = 1.
\]
By multiplying $(\bm{\Gamma}^{-1})\bm{e}$ from the right on both sides
of this equation and rearranging terms, we can verify that 
\[
\bm{\pi}(-\bm{\Gamma}\bm{e})
= 
\frac{1}{\bm{\gamma} (-\bm{\Gamma})^{-1}\bm{e}}
=
\lambda,
\]
where the last equality follows from
(\ref{eq:EG-PH}) and $\lambda = 1/\E[G]$.

\yoshiaki{

\begin{remark}
\label{remark:MAP}

In this model, the total input process 
$(X_{\total}(t))_{t \geq 0} := (X(t) + X_{\bg}(t))_{t \geq 0}$
follows a compound \textit{Markovian arrival process (MAP)}.
To be more specific, let $\bm{C}$ and $\bm{D}(x)$ (for $x \geq 0$) denote
$M \times M$ matrices defined as 
\begin{equation}
\bm{C} = \bm{\Gamma} - \lambda_{\bg} \bm{I},
\quad
\bm{D}(x) 
= 
(-\bm{\Gamma}\bm{e})\bm{\gamma} F_H(x) 
+ \lambda_{\bg} F_{H_{\bg}}(x)\bm{I},
\label{eq:C-D(x)-def}
\end{equation}
where $\bm{I}$ denotes an $M \times M$ unit matrix.
Let $S(t) \in \{1,2,\ldots,M\}$ denote the phase of the system at time
$t$, i.e., the state of the phase-type distribution of the \michel{inter-generation} time.
We define $\bm{x}_{\total}(t, x)$ as a $1 \times M$ vector whose
$j$th element $j=1,2,\ldots,M$ is given by
\[
[\bm{x}_{\total}(t)]_j = \Pr(X(t) \leq x, S(t) = j).
\]
\michel{Then} the vector $(X_{\total}(t), S(t))_{t \geq 0}$ forms a bivariate Markov process, and its 
transition law is completely determined by $\bm{C}$ and $\bm{D}(\cdot)$: as $\varDelta t \downarrow 0+$,
\[
\bm{x}_{\total}(t+\varDelta t, x)
= 
\bm{x}_{\total}(t, x) \,(\bm{I}+\bm{C} \varDelta t)
+
\int_0^x \bm{x}_{\total}(t, x-y) \,\dd \bm{D}(y) \,\varDelta t
+o(\varDelta t). 
\]
The first and second term on the right-hand side of this
equation correspond \michel{to} respectively the state-transitions without and
with an upward jump in $X_{\total}(t)$.

\end{remark}
}

\yoshiaki{
Recall that} evaluating the expression for $\E[A]$, as given in (\ref{eq:EA-by-q(y)-0}), requires Laplace inversion
as well as integration over $y\in[0,\infty)$.
It is known that the accuracy of Laplace-inversion methods is highly
dependent on the shape of the function under consideration, which
sometimes leads to inconsistent numerical results between different
inversion algorithms \cite{Davies1979}.
This led us to develop an alternative approach that utilizes our
assumption of phase-type distributed \michel{inter-generation} times.

The key idea in this alternative approach lies in expanding $f_G(\cdot)$, as given by (\ref{eq:f_G-PH}), as an
Poisson-weighted average of non-negative numbers (cf.\ the uniformization
technique \cite[P. 154]{Tijms1994}); the favorable numerical properties of the resulting computation scheme are due to the fact that the weights used all have the same sign. 

\yoshiaki{
We start by presenting \michel{a convenient representation of ${\mathbb E}[A]$}, whose proof is provided
in Appendix \ref{appendix:EA-formula-PH}.

\begin{lemma}
\label{lemma:EA-formula-PH}

The mean AoI $\E[A]$ can be expressed as
\begin{equation}
\E[A]
=
\frac{\E[H]}{\mu}
+
\sum_{k=0}^{\infty}
\frac{c_k}{\theta\E[G]}
\left\{
\frac{(k+1)(k+2)}{2\theta^2}
+
(k+1) q_{k+1}
\right\},
\label{eq:EA-formula-PH}
\end{equation}
where $\theta$ denotes the maximum absolute values of
diagonal elements of $\bm{\Gamma}$
\[
\theta := \max_i |[\bm{\Gamma}]_{i,i}|,
\]
and $c_k$ ($k=0,1,\ldots$) and $q_k$ ($k=0,1,\ldots$) denote
non-negative numbers given by
\begin{align}
c_k 
&:= 
\bm{\gamma} (\bm{I}+\theta^{-1}\bm{\Gamma})^k (-\bm{\Gamma})\,\bm{e},
\quad
k = 0,1,\ldots,
\label{eq:c-uniform}
\\
q_k &:= \int_0^{\infty}
q(y) 
\cdot \frac{e^{-\theta y}(\theta y)^k}{k!}\,
\dd y,
\quad
k = 0,1,\ldots.
\label{eq:q_k-def}
\end{align}
\end{lemma}
}


\yoshiaki{
\michel{Lemma \ref{lemma:EA-formula-PH} entails that} the mean AoI $\E[A]$ \michel{can be expressed} in terms of $q_k$ ($k=0,1,
\ldots$) defined as (\ref{eq:q_k-def}). 
Recall that, as shown in (\ref{eq:q-t-inv}),
\[
\int_0^{\infty} q(y)e^{-\omega y} \, \dd y
=
-\frac{1-\rho_{\bg}}{\omega^2}
+
\frac{1}{\omega}
\left(
\frac{f_D^*(\psi(\omega))}{\psi(\omega)}
+\E[D]
\right).
\]
This Laplace transform is related to $q_k$ ($k=0,1,\ldots$) through 
its $z$-transform $q^*(z)$ (for $|z| < 1$):
\begin{align}
q^*(z) 
:= 
\sum_{k=0}^{\infty} q_k z^k
&=
\int_0^{\infty}
q(y) e^{-(\theta-\theta z)y}
\dd y
\nonumber
\\
&=
-\frac{1-\rho_{\bg}}{(\theta - \theta z)^2}
+
\frac{1}{\theta-\theta z}
\left(
\frac{f_D^*(\psi(\theta - \theta z))}{\psi(\theta - \theta z)}
+\E[D]
\right),
\label{eq:q^*-by-psi}
\end{align}
where the second equality follows from (\ref{eq:q_k-def}).
\michel{This means that} $q_k$ ($k=0,1,\ldots$) is obtained as the $k$th coefficient of
the power series expansion of (\ref{eq:q^*-by-psi}) with respect to $z$.

\michel{Bearing in mind the form of \eqref{eq:q^*-by-psi}, the next step is to obtain expressions for $f_D^*(\cdot)$ and $\psi(\cdot)$ evaluated in the relevant arguments. This concretely means that, in order to determine} the coefficients $q_k$, the key quantities are 
the coefficients $d^{(m)}$ ($m=0,1,\ldots$) and $b^{(i)}$ ($i=0,1,
\ldots$) defined via
\begin{align}
f_D^*(\zeta - \zeta z) =& \sum_{m=0}^{\infty} d^{(m)} z^m,
\label{eq:d-m-def}
\\
b(z) :=&
\frac{\theta}{\theta + \lambda_{\bg}} \cdot z 
+\frac{\lambda_{\bg}}{\theta + \lambda_{\bg}} \cdot f_{B_{\bg}}^*(\theta - \theta z)
\nonumber
\\
=&
\sum_{i=0}^{\infty} b^{(i)} z^i,
\label{eq:b(z)-def}
\end{align}
where $\zeta$ is defined as
\begin{equation}
\zeta := \theta + \lambda_{\bg}.
\label{eq:theta-zeta}
\end{equation}
Recall that $f_{B_{\bg}}^*(s)$ is given by the solution of (\ref{eq:f_B_bg-def}). 
By definition (\ref{eq:psi-def}) of $\psi(\omega)$, we can verify that
$\psi(\theta - \theta z)$ and $f_D^*\bigl(\psi(\theta - \theta z)\bigr)
$ on the right-hand side of (\ref{eq:q^*-by-psi}) are represented as
\begin{align}
\psi(\theta - \theta z) &= 
\theta - \theta z + \lambda_{\bg} - \lambda_{\bg}f_{B_{\bg}}^*(\theta - \theta z)
\nonumber
\\
&=
\theta + \lambda_{\bg} 
- (\theta + \lambda_{\bg})
\cdot
\left\{
\frac{\theta}{\theta + \lambda_{\bg}} \cdot z 
+ \frac{\lambda_{\bg}}{\theta + \lambda_{\bg}} \cdot f_{B_{\bg}}^*(\theta - \theta z)
\right\}
\nonumber
\\
&=
\zeta - \zeta b(z)
\label{eq:psi-by-b}
=
\zeta - \zeta \sum_{i=0}^{\infty} b^{(i)} z^i,
\\
f_D^*\bigl(\psi(\theta - \theta z)\bigr) 
&= 
\sum_{m=0}^{\infty} d^{(m)}\psi(\theta - \theta z)
=
\sum_{m=0}^{\infty} 
d^{(m)} 
\cdot
\Bigl\{
\zeta - \zeta \sum_{i=0}^{\infty} b^{(i)} z^i
\Bigr\}.
\nonumber
\end{align}
Furthermore, the mean system delay $\E[D]$ is given by
\begin{equation}
\E[D] = \frac{1}{\zeta}\sum_{m=0}^{\infty} md^{(m)},
\label{eq:ED-PH}
\end{equation}
which is obvious from the relation
\[
\E[D] = (-1) \cdot \lim_{s \to 0+} \frac{\dd}{\dd s} f_D^*(s)
= 
\frac{1}{\zeta}\lim_{z \to 1} \frac{\dd}{\dd z} f_D^*(\zeta - \zeta
z).
\]
Therefore, once $d^{(m)}$ ($m=0,1,\ldots$) and $b^{(i)}$ ($i=0,1,
\ldots$) are obtained, the coefficients $q_k$ ($k=0,1,\ldots$) in
(\ref{eq:EA-formula-PH}) are determined from (\ref{eq:q^*-by-psi}) in a straightforward way. \michel{In Lemma \ref{lemma:b-coeff} we develop a computationally sound scheme for evaluating
the $b^{(i)}$, while Lemma \ref{lemma:d-m} provides the $d^{(m)}$.}

\begin{remark}
For any non-negative random variable $X$ with LST $f_X^*(s)$, 
we have the following expression for any $\eta > 0$:
\[
f_X^*(\eta - \eta z)
=
\int_0^{\infty} e^{-(\eta - \eta z)x}\, \dd F_X(x)
=
\sum_{m=0}^{\infty} 
z^m
\int_0^{\infty} \frac{e^{-\eta x} (\eta x)^m}{m!} \, \dd F_X(x)
=
\sum_{m=0}^{\infty} 
x^{(m)}
z^m,
\]
where
\[
x^{(m)} := \int_0^{\infty} \frac{e^{-\eta x} (\eta x)^m}{m!} \, \dd
F_X(x),
\quad
m = 0,1,\ldots.
\]
Therefore, $x^{(m)}$ represents the probability that a Poisson process
of rate $\eta$ generates $m$ arrivals in a random interval of length
$X$.

From this observation, we can verify that $d^{(m)}$ ($m=0,1,\ldots$)
and $b^{(i)}$ ($i=0,1, \ldots$) defined in (\ref{eq:d-m-def}) and
(\ref{eq:b(z)-def}) are proper probability mass functions.
\end{remark}

\michel{We can compute $d^{(m)}$ ($m=0,1,\ldots$) and $b^{(i)}$ ($i=0,1, \ldots$) in a numerically stable way, as follows.}
As we will see, they are given in terms of \michel{the} probability mass functions 
$\hat{h}^{(m)}$ and $\hat{h}_{\bg}^{(m)}$ ($m=0,1,\ldots$) \michel{that are defined as}
\begin{equation}
\hat{h}^{(m)} 
:= 
\int_0^{\infty} 
\frac{e^{-\zeta x/\mu}(\zeta x/\mu)^m }{m!}
\dd F_H(x),
\quad
\hat{h}_{\bg}^{(m)} 
:= 
\int_0^{\infty} 
\frac{e^{-\zeta x/\mu}(\zeta x/\mu)^m }{m!}
\dd F_{H_{\bg}}(x).
\label{eq:h-uniform}
\end{equation}
It is noted that $\hat{h}^{(m)}$ (and $\hat{h}_{\bg}^{(m)}$,
respectively) represents the probability that a Poisson process of rate $\zeta$ generates 
$m$ arrivals in \michel{an interval of random length $H/\mu$} ($H_{\bg}/\mu$, respectively). 
\michel{As a consequence, for given specific
distributions of the random variables $H$ and $H_{\rm bg}$, the computation of $\hat{h}^{(m)}$ and
$\hat{h}_{\bg}^{(m)}$ is straightforward.}

The sequences $b^{(i)}$ ($i=0,1,\ldots$) and $d^{(m)}$ ($m=0,1,\ldots$)
are determined by Lemmas~\ref{lemma:b-coeff} and \ref{lemma:d-m} below. 
Their proofs are provided in Appendices \ref{appendix:b-coeff} and \ref{appendix:d-m}. 
\begin{lemma}
\label{lemma:b-coeff}

$b^{(i)}$ ($i = 0,1,\ldots$) is given by
\begin{equation}
b^{(i)} 
= 
\frac{\id{i=1}\theta}{\theta+\lambda_{\bg}}
+
\frac{\lambda_{\bg}}{\theta+\lambda_{\bg}}
\sum_{k=i}^{\infty} 
\frac{1}{k-i+1}
\comb{k}{i}
\left(\frac{\theta}{\theta + \lambda_{\bg}}\right)^i
\left(\frac{\lambda_{\bg}}{\theta + \lambda_{\bg}}\right)^{k-i}
y_{k-i+1}^{(k)},
\label{eq:b^i-def}
\end{equation}
where $y_k^{(m)}$ given recursively by
\begin{align}
y_1^{(m)} &= \hat{h}_{\bg}^{(m)},
\quad
m = 0,1,\ldots,
\label{eq:y_k-recursion-1}
\\
y_k^{(m)} &= \sum_{i=0}^m y_{k-1}^{(i)} \hat{h}_{\bg}^{(m-i)},
\quad
k = 1,2,\ldots,\,
m = 0,1,\ldots,
\label{eq:y_k-recursion-k}
\end{align}
with $\hat{h}_{\bg}^{(m)}$ defined as (\ref{eq:h-uniform}).
\end{lemma}
\begin{lemma}
\label{lemma:d-m}

Let $(\bm{Q}^{(n)})_{n=0,1,\ldots}$ denote a sequence of $M \times M$
matrices defined as
\begin{equation}
\bm{Q}^{(0)} = \bm{C},
\quad
\bm{Q}^{(n)}
=
\bm{C} +
\sum_{k=0}^{\infty} 
\hat{\bm{D}}^{(k)}
\cdot
(\bm{I}+\zeta^{-1}\bm{Q}^{(n-1)})^k,
\;\;
n = 1,2,\ldots,
\label{eq:Q-recursion-uniform}
\end{equation}
where $\bm{C}$ is given by (\ref{eq:C-D(x)-def}) and 
$\hat{\bm{D}}^{(m)}$ ($m=0,1,\ldots$) is defined as
\begin{equation}
\hat{\bm{D}}^{(m)}
=
(-\bm{\Gamma}\bm{e})\bm{\gamma}
\hat{h}^{(m)}
+ \bm{I}\lambda_{\bg}\hat{h}_{\bg}^{(m)}.
\label{eq:hatD-m}
\end{equation}
$(\bm{Q}^{(n)})_{n=0,1,\ldots}$ is an elementwise non-decreasing sequence
and its limit $\bm{Q} := \lim_{n \to \infty} \bm{Q}^{(n)}$ represents a
proper transition rate matrix of an irreducible continuous-time Markov
chain with $M$ states. Let $\bm{\kappa}$ denote the stationary
probability vector associated with $\bm{Q}$:
\begin{equation}
\bm{\kappa}\bm{Q} = \bm{0},
\quad
\bm{\kappa}\bm{e} = 1.
\label{eq:kappa-def}
\end{equation}
The coefficients $d^{(m)}$ ($m=0,1,\ldots$) of $d^*(\zeta-\zeta z)$ are then given by 
\begin{equation}
d^{(m)} 
= 
\frac{1}{\lambda}
\sum_{k=0}^m \hat{\bm{v}}^{(k)}
(-\bm{\Gamma}\bm{e})\hat{h}^{(m-k)},
\quad
m =0,1,\ldots,
\label{eq:d(m)-PH}
\end{equation}
where $\hat{\bm{v}}^{(m)}$ is given by
\begin{align}
\hat{\bm{v}}^{(0)} 
&= (1-\rho-\rho_{\bg})\bm{\kappa}
(\bm{I}-\bm{R}^{(0)})^{-1},
\label{eq:v-uniform-0}
\\
\hat{\bm{v}}^{(m)} 
&= \left(\sum_{k=0}^{m-1} \hat{\bm{v}}^{(k)}\bm{R}^{(m-k)}\right)
(\bm{I}-\bm{R}^{(0)})^{-1},
\quad
m = 1,2,\ldots,
\label{eq:v-uniform-m}
\end{align}
with $\bm{R}^{(m)}$ ($m=0,1,\ldots$) defined as
\begin{equation}
\bm{R}^{(m)} 
= 
\zeta^{-1}
\sum_{k=0}^{\infty}
\hat{\bm{D}}^{(m+k+1)}
(\bm{I}+\zeta^{-1}\bm{Q})^k.
\label{eq:R-uniform}
\end{equation}
\end{lemma}

\begin{remark}
\label{remark:IR-inverse_IQ_subst}

The following facts are shown in Appendix \ref{appendix:d-m} (see
Remark \ref{remark:IR-inverse_IQ_subst_proof}), which
ensures the well-behavedness of the computational algorithm to be
developed. \michel{More specifically,}
\begin{itemize}
\item[(i)] $\bm{I}+\zeta^{-1}\bm{Q}^{(n)}$ ($n=0,1,\ldots$) is
substochastic, i.e., it has non-negative elements with row sums not exceeding one, \michel{and} 
\item[(ii)] the matrix inverse $(\bm{I}-\bm{R}^{(0)})^{-1}$ exists.
\end{itemize}
\end{remark}
}

\yoshiaki{
We are now in a position to present a (numerically stable) computation
scheme for the mean AoI $\E[A]$; the proof of the following theorem is
provided in Appendix \ref{appendix:meanA-PH}.
\begin{theorem}\label{thm_meanA}
The mean AoI is given by (\ref{eq:EA-formula-PH}) with $q_k$
($k = 0,1,\ldots$) determined recursively by 
\begin{align}
q_0 
&= 
-\frac{1-\rho_{\bg}}{\theta^2}
+
\frac{\E[D]}{\theta} + \frac{1}{\theta\zeta}
\sum_{\ell=0}^{\infty} b_{\ell}^{(0)}\sum_{m=0}^{\ell}d^{(m)},
\label{eq:q-recursion-0}
\\
q_k
&=
q_{k-1}
-\frac{1-\rho_{\bg}}{\theta^2}
+ 
\frac{1}{\theta\zeta}
\sum_{\ell=0}^{\infty} 
b_{\ell}^{(k)}\sum_{m=0}^{\ell}d^{(m)},
\quad
k = 1,2,\ldots,
\label{eq:q-recursion-k}
\end{align}
where $d^{(m)}$ ($m = 0,1,\ldots$) is given in Lemma \ref{lemma:d-m},
and $\E[D]$ is given by (\ref{eq:ED-PH}).
Also, $b_{\ell}^{(k)}$ (for $\ell = 0,1,\ldots$, $k=0,1,\ldots$) is
defined as
\begin{equation}
b_0^{(k)} = \id{k=0},
\quad
b_{\ell}^{(k)} = \sum_{i=0}^k b_{\ell-1}^{(k-i)} b^{(i)},
\;\;
\ell = 1,2,\ldots,
\label{eq:b_ell-recursion}
\end{equation}
with $b^{(i)}$ ($i=0,1,\ldots$) given in Lemma \ref{lemma:b-coeff}.
\end{theorem}
}

In Figure \ref{figure:algorithm}, we summarize the resulting computational
algorithm for $\E[A]$.
\yoshiaki{An implementation of the algorithm is publicly available at \cite{github}.}
We note that in the computations it is necessary to truncate some infinite series; cf.\ (\ref{eq:q-recursion-0}) and (\ref{eq:q-recursion-k}).
To check the numerical precision of the procedure, the interpretation
(\ref{eq:q-as-mean}) of $q(y)$ as the
mean transient workload in the M/GI/1 queue is useful.
More specifically, we can verify from (\ref{eq:q_k-def}) that 
\begin{align*}
\lim_{k \to \infty} q_k 
&= 
\frac{1}{\theta}
\cdot
\lim_{k \to \infty} 
\int_0^{\infty}
q(y) 
\cdot \frac{e^{-\theta y}(\theta y)^k \theta}{k!}\,
\dd y
\\
&=
\frac{1}{\theta}
\cdot
\lim_{k \to \infty}
\E\left[
q(E_{k+1})
\right]
\\
&= 
\frac{1}{\theta}
\cdot
\lim_{y \to \infty} q(y),
\end{align*}
where $E_k$ denotes a random variable following a $k$-stage Erlang distribution
with rate $\theta$.
Because $q(y)$ represents the transient virtual waiting time in the
M/GI/1 queue with only the background stream (cf.\ (\ref{eq:q-as-mean})),
it tends to the stationary mean virtual waiting time in the M/GI/1
queue as $y \to \infty$. We thus have
\begin{equation}\label{eq:precloss}
\lim_{k \to \infty} q_k 
=
\frac{1}{\theta}
\cdot
\frac{\lambda_{\bg}\E[H_{\bg}^2]}{2(1-\rho_{\bg})\mu^2}.
\end{equation}
Therefore, for large values of $k$, the difference between the computed $q_k$ and the right-hand side of \eqref{eq:precloss} serves as a
straightforward indication of the loss in numerical precision due to the
truncation.

\begin{figure}[tp]
\centering
\fbox{%
\begin{minipage}{0.95\textwidth}
\begin{itemize}
\item[\hbox to 2.2ex{Input:\hss}] 
\hspace*{3.0ex}
$\lambda_{\bg}$, $F_{H_{\bg}}(\cdot)$, $(\bm{\gamma}, \bm{\Gamma})$,
$F_H(\cdot)$, $\mu$
\item[\hbox to 2.2ex{Output:\hss}] 
\hspace*{3.0ex}
$\E[A]$
\item[(a)]
Compute $\bm{C}$ with (\ref{eq:C-D(x)-def}), 
$c_k$ ($k = 0,1,\ldots$) with (\ref{eq:c-uniform}),
and $\hat{h}^{(m)}$  and $\hat{h}_{\bg}^{(m)}$
($m = 0,1,\ldots$) with (\ref{eq:h-uniform}).

\item[(b)]
Compute $\bm{Q}$ by letting $\bm{Q}^{(0)} = \bm{C}$ and iterating
(\ref{eq:Q-recursion-uniform}) until convergence.

\item[(c)]
Compute $\bm{R}^{(m)}$ ($m = 0,1,\ldots$) with (\ref{eq:R-uniform})
and compute $\hat{\bm{v}}^{(m)}$ ($m = 0,1,\ldots$) with
(\ref{eq:v-uniform-0}) and (\ref{eq:v-uniform-m}).
\yoshiaki{Also, compute $\bm{\kappa}$ satisfying (\ref{eq:kappa-def}).}

\item[(d)]
Compute $d^{(m)}$ ($m = 0,1,\ldots$) with (\ref{eq:d(m)-PH}) and
compute $\E[D]$ with (\ref{eq:ED-PH}).

\item[(e)]
Compute $b^{(i)}$ ($i = 0,1,\ldots$) with (\ref{eq:y_k-recursion-1}),
(\ref{eq:y_k-recursion-k}), and (\ref{eq:b^i-def}).

\item[(f)]
Compute $\E[A]$ with (\ref{eq:EA-formula-PH}),
(\ref{eq:q-recursion-0}), and (\ref{eq:q-recursion-k}).
\end{itemize}\end{minipage}}
\caption{The algorithm for computing the mean AoI $\E[A]$ in the PH+M/GI+GI/1 queue.}
\label{figure:algorithm}
\end{figure}

\section{Numerical experiments}\label{sec:num}

\michel{In this section we provide a sequence of numerical experiments that systematically assess the impact of the model parameters on the expectation of the AoI. In particular, we will assess the precision of the bounds on ${\mathbb E}[A]$ that were given in Theorem \ref{theorem:A-bound-mean}, where the benchmark is computed, for the class of PH+M/GI+GI/1 systems, by applying Theorem \ref{thm_meanA}. The main conclusion is that the upper bound is particularly tight. This insight leads to a closed-form expression for the near-optimal \yoshiaki{generation} rate $\lambda$.}

We work with the setup of Section~\ref{sec:PH} where the \michel{inter-generation} times
$(G_n)_{n=1,2,\ldots}$ follow a phase-type distribution. \michel{As we are in the PH-M/GI/GI/1 setting, we can use the algorithm given in Figure \ref{figure:algorithm} to numerically evaluate ${\mathbb E}[A].$}
The idea is that we rely on the two-moments fit advocated in \cite{Tijms1994}, which identifies a convenient low-dimensional phase-type distribution with a given mean and variance.
With the coefficient of variation of the random variable $G$ defined as
\[s_G = \sqrt{\frac{{\mathbb V}{\rm ar}[G]}{({\mathbb E}[G])^2}},\]
it distinguishes between the cases $s_G\le 1$ and $s_G>1.$ 
In the former case, we fit a mixture of two Erlang distributions: with probability $p$ the fitted distribution is of Erlang($k,\nu$) type, and with probability $1-p$ of Erlang($k+1,\nu$) type. 
The integer-valued parameter $k$ is chosen such that $s_G$ lies in the interval between $1/(k+1)$ and $1/k$, and then the parameters $p$ and $\nu$ are chosen such that the mixed Erlang random variable has the given mean and variance. In the latter case, we fit a mixture of two Exponential distributions: with probability $p$ the fitted distribution is exponential with parameter $\nu_1$, and with probability $1-p$ exponential with parameter $\nu_2$. 
Under the additional constraint of `balanced means', one can uniquely pick the parameters $p$, $\nu_1$, and $\nu_2$ so that the resulting random variable has the given mean and variance.

In a first experiment we study the effect of the `variability' of the \michel{inter-generation} times. 
Through this experiment we can quantify the error we would make by wrongly assuming that these \michel{inter-generation} times are exponentially distributed. 
We have picked the following parameters.
The \yoshiaki{generation} rate $\lambda$ of the tagged stream, which is the reciprocal of the mean \michel{inter-generation} time ${\mathbb E}[G]$, was put equal to $0.05$.
The (Poissonian) arrival rate $\lambda_{\rm bg}$ corresponding to the background traffic equals $0.8$, $0.85$, and $0.9$.
The service times of the tagged stream as well as the background stream have been chosen deterministic of size 1, and time is normalized such that the service rate $\mu$ equals $1$.
Figure \ref{Fig_E1} shows the mean AoI as a function of the coefficient of variation $s_G$ pertaining to the \michel{inter-generation} times. 
As expected, ${\mathbb E}[A]$ is an increasing function of $s_G$: more variability typically leads to longer queues.
The more important conclusion, however, is that misspecifying our GI+M/GI+GI/1 model by its more elementary M+M/GI+GI/1 counterpart, that was covered by \cite{Moltafet2020,Inoue2024}, may lead to significant errors.
Indeed, for the parameters considered, the difference between ${\mathbb E}[A]$ in the M+M/GI+GI/1 case (i.e., the special case of the GI+M/GI+GI/1 system in which $s_G=1$) and the GI+M/GI+GI/1 can already be a factor~2. 

From Figure \ref{Fig_E1}, we also observe that these three curves with
different $\rho_{\bg}$ appear to be in parallel with each other, i.e.,
the variability $s_G$ of the \michel{inter-generation} times has a common impact
on $\E[A]$ regardless of the background traffic intensity. Notice that
the mean AoI is approximated by the sum of the
mean forward recurrence time $\E[G^2]/(2\E[G]) = (1+s_G^2)\E[G]/2$ of
\michel{inter-generation} times and the mean system delay; cf.\ \cite[Equation (2)]{Inoue2019}, and also
\eqref{eq:EA-lower}--\eqref{eq:EA-upper-W_star} in the present paper. 
Therefore, the effect of $s_G$ on $\E[A]$ is roughly evaluated as
$(1+s_G^2)\E[G]/2$; this is indeed verified in Figure
\ref{Fig_diff_CvG}, where the `residual component' $\E[A]- (1+s_G^2)\E[G]/2$ as a
function of $s_G$ is shown to be almost constant.

\begin{figure}[tp]
\centering
\includegraphics[scale=1.0]{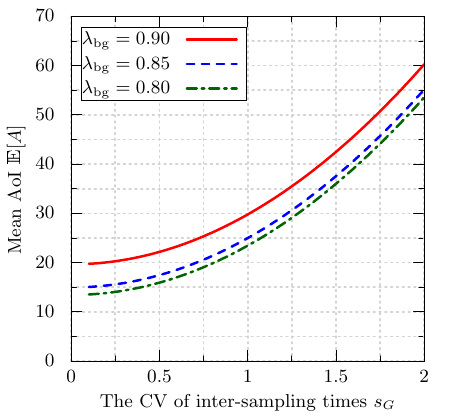}
\caption{\label{Fig_E1}The mean AoI, as a function of the coefficient of variation (CV) of the \michel{inter-generation} times $s_G$.}
\mbox{}
\\
\includegraphics[scale=1.0]{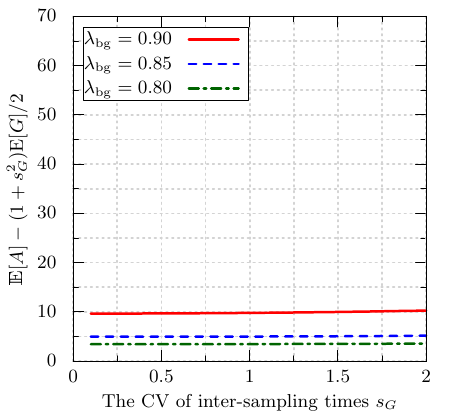}
\caption{\label{Fig_diff_CvG}The AoI value of Figure \ref{Fig_E1} subtracted by the mean forward-recurrence time of the \michel{inter-generation} times.}
\end{figure}

In the second experiment, we pick the parameters as before, but now
with $s_G$ being fixed at $0.3$ and varying $\lambda$ \yoshiaki{and the coefficient of variation $s_H$ of service times}. The result is
shown in Figure \ref{Fig_E2}, where the explicit bounds derived in
Theorem~\ref{theorem:A-bound-mean} are also plotted.
We observe a pattern that was found for more elementary AoI models: the mean AoI first decreases in $\lambda$, and later increases. 
The explanation for the phenomenon is that small values of $\lambda$ mean that the monitor is provided with relatively little information, leading to untimely knowledge and hence a large AoI, whereas for large values of $\lambda$ there will be more queueing delay and hence also a large AoI.
The shape of the graph interestingly means that there is an optimal value of $\lambda$, i.e., a value of $\lambda$ that minimizes ${\mathbb E}[A].$
In this experiment we also observe that the optimal $\lambda$ decreases in $\lambda_{\rm bg}$.
To understand this, recall that the mean AoI is approximated by the
sum of the mean forward recurrence time of \michel{inter-generation} times and the mean system delay as mentioned above. Then observe that the former is a decreasing function of the \yoshiaki{generation} rate $\lambda$, whereas the latter is an increasing function of $\lambda$.
For smaller values of the rate $\lambda_{\rm bg}$, the system delay significantly increases from a smaller value of $\lambda$ on, entailing that the optimal $\lambda$ becomes smaller.

\michel{The most important conclusion from Figure \ref{Fig_E2} is that the explicit upper bound accurately follows the exact curve of the mean AoI $\E[A]$ as a function of the
\yoshiaki{generation} rate $\lambda$. In particular, the value of $\lambda$ that minimizes the upper bound is a precise approximation of the value of $\lambda$ that minimizes the mean AoI $\E[A]$. Interestingly, the optimizer of the upper bound can be given in closed form.  Indeed, replacing  ${\mathbb E}[G]$ by $\lambda^{-1}$ and ${\mathbb E}[G^2]$ by $(s_G^2+1)\,\lambda^{-2}$ (so as to enforce that $G$ has mean $\lambda^{-1}$ and coefficient of variation $s_G$) in the upper bound of Theorem \ref{theorem:A-bound-mean}, the rate 
$\lambda^\star$ that optimizes this upper bound is the minimizer of ${\mathbb E}^+_\lambda[A]$, where 
\begin{equation}\label{eq:foc}{\mathbb E}^+_\lambda[A]:=\frac{\lambda{\mathbb E}[{H^2}]+\lambda_{\bg}{\mathbb E}[H_{\bg}^2]}{2\mu(\mu-\lambda{\mathbb E}[H]-\lambda_{\bg}{\mathbb E}[H_{\bg}])}
+ 
\frac{{\mathbb E}[H]}{\mu}
+ 
\frac{s_G^2+1}{2\lambda}.\end{equation}
Observing that ${\mathbb E}^+_\lambda[A]$ is a convex function of $\lambda$, we can obtain $\lambda^\star$ by equating its first derivative to zero.
Taking the first derivative of the expression displayed in the right-hand side of \eqref{eq:foc}, 
\[\frac{1}{2\mu}\frac{{\mathbb E}[H^2]\,(\mu-\lambda_{\bg}{\mathbb E}[H_{\bg}])+\lambda_{\bg}{\mathbb E}[H_{\bg}^2]\,{\mathbb E}[H]}{(\mu-\lambda{\mathbb E}[H]-\lambda_{\bg}{\mathbb E}[H_{\bg}])^2}- \frac{s_G^2+1}{2\lambda^2}.\]
After some elementary calculations we thus find that $\lambda^\star$ satisfies
\begin{equation}\frac{\mu-\lambda{\mathbb E}[H]-\lambda_{\bg}{\mathbb E}[H_{\bg}]}{\Omega}=\frac{\lambda}{\sigma_G},\label{eq:firstordercondition}
\end{equation}
where we denote $\sigma_G:=\sqrt{s_G^2+1}$ and 
\[\Omega :=\sqrt{{\mathbb E}[H^2]+\big({\mathbb E}[H]{\mathbb E}[H_\bg^2]-{\mathbb E}[H^2]{\mathbb E}[H_\bg]\big)\,\lambda_\bg/\mu};\]
note that the expression under the root-sign is positive due to $\rho_{\rm bg}<1.$
 After again a few straightforward steps, we conclude that
\begin{equation}\lambda^\star =\frac{\mu(1-\rho_\bg)\sigma_G}{\Omega\,+\,{\mathbb E}[H]\sigma_G}.\label{eq:optimallambda}\end{equation}
This expression provides us with various useful quantitative insights that are highly relevant from a practical standpoint: it for instance reveals how $\lambda^\star$ decreases as $\varrho_\bg\uparrow 1$ and how $\lambda^\star$ increases in $\sigma_G$ (and hence in $s_G$). 
As a next step, we can insert $\lambda^\star$ into the upper bound to get insight into the mean AoI as a function of the model parameters. Applying that $\lambda^\star$ satisfies \eqref{eq:firstordercondition}, we directly find that our optimal upper bound reduces to
\begin{align*}
{\mathbb E}^+_{\lambda^\star}[A]  &=\frac{\lambda^\star{\mathbb E}[{H^2}]+\lambda_{\bg}{\mathbb E}[H_{\bg}^2]}{2\mu}\frac{\sigma_G}{\lambda^\star \Omega}+ 
\frac{{\mathbb E}[H]}{\mu}
+ 
\frac{\sigma_G^2}{2\lambda^\star}
\\
&=\frac{\E[H^2]\sigma_G}{2\mu\Omega}+\frac{{\mathbb E}[H]}{\mu}+\frac{1}{\lambda^\star}\left(\frac{\lambda_{\rm bg}{\mathbb E}[H_{\rm bg}^2]\sigma_G}{2\mu\Omega}+\frac{\sigma^2_G}{2}\right)
\\
&=\frac{\E[H^2]\sigma_G}{2\mu\Omega}+\frac{{\mathbb E}[H]}{\mu}+\frac{\Omega\,+\,{\mathbb E}[H]\sigma_G}{\mu(1-\rho_\bg)\sigma_G}\left(\frac{\lambda_{\rm bg}{\mathbb E}[H_{\rm bg}^2]\sigma_G}{2\mu\Omega}+\frac{\sigma^2_G}{2}\right),
\end{align*}
where in the last step we plugged in the expression for $\lambda^\star$ identified in \eqref{eq:optimallambda}. After elementary but tedious computations, it turns out that this expression simplifies to
\begin{equation}
{\mathbb E}^+_{\lambda^\star}[A]  = \frac{2\mu\Omega\sigma_G+2(\mu-\lambda_{\rm bg}\E[H_{\rm bg}]){\mathbb E}[H]+\lambda_\bg \E[H_{\rm bg}^2]+\mu\E[H]\sigma^2_G}{2\mu(\mu-\lambda_{\rm bg}\E[H_{\rm bg}])}.
\label{eq:EA-star}
\end{equation}
Figure \ref{Fig_E2} suggests that the AoI minimizing $\lambda$, to be denoted by $\lambda_{\rm opt}$, is close to $\lambda^*$ (i.e., the $\lambda$ that minimizes the upper bound); in the figure the points $(\lambda^\star,{\mathbb E}^+_{\lambda^\star}[A])$ have been made visible by stars ($\star$). Table \ref{tab:EA} shows that, across a set of representative examples, the expected AoI corresponding to $\lambda_{\rm opt}$ is just a tiny fraction lower than the expected AoI corresponding to $\lambda^\star$, as given by \eqref{eq:EA-star}; the difference is typically well below $0.1\%$. This justifies the use of $\lambda^\star$, as given by \eqref{eq:optimallambda}, as a proxy for the AoI minimizing generation rate.}

It is further noted from Figure \ref{Fig_E2} that the lower bound \michel{is a decreasing function of $\lambda$ which}
significantly underestimates
$\E[A]$; this is in particular the case when $\lambda$ has a large impact on the system congestion, 
i.e., in a situation where $\rho$ is close to the
`remaining capacity' $1-\rho_{\bg}$ of the system. 
To obtain further insights related to this perspective, we set
$\lambda=0.05$ and scale up $\lambda_{\bg}$ and $\mu$ simultaneously
while keeping $\rho_{\bg}$ fixed, where the other parameters are the same as those underlying Figure~\ref{Fig_E2}. Figure~\ref{Fig_E3} plots the exact value of $\E[A]$ as well as the corresponding bounds, as functions of the occupancy $\rho_{\bg}
/(\rho+\rho_{\bg})$ of the background traffic. 
From this figure, we observe that the explicit bounds in Theorem~\ref{theorem:A-bound-mean} are tight when the background traffic is
dominant and the single sensor accounts for at most 1\% of the traffic; 
in practice, modern broadband communication systems simultaneously serve vast
numbers of data flows, thus resulting in the occupancy
of a single sensor typically dropping far below 1\%. In this sense, Theorem~\ref{theorem:A-bound-mean} provides a simple yet highly accurate measure of
`information freshness' for monitoring systems with coexisting background
streams.

\begin{figure}
\centering
\includegraphics[scale=1.0]{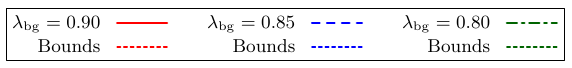}
\\[1ex]
\begin{minipage}[t]{0.49\textwidth}
\centering
\includegraphics[scale=1.0]{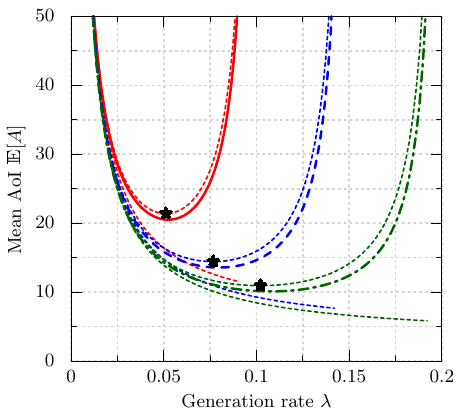}
\mbox{}
\\[-5ex]
(a) $s_H = 0$
\end{minipage}
\begin{minipage}[t]{0.49\textwidth}
\centering
\includegraphics[scale=1.0]{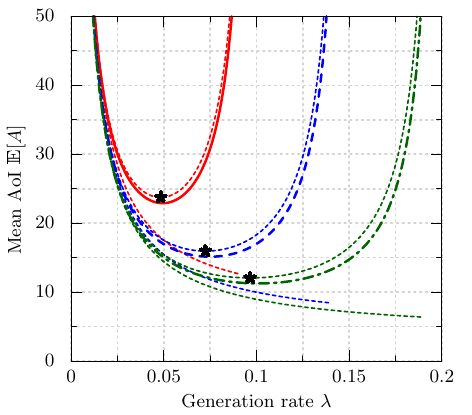}
\mbox{}
\\[-5ex]
(b) $s_H = 0.5$
\end{minipage}
\\[3ex]
\begin{minipage}[t]{0.49\textwidth}
\centering
\includegraphics[scale=1.0]{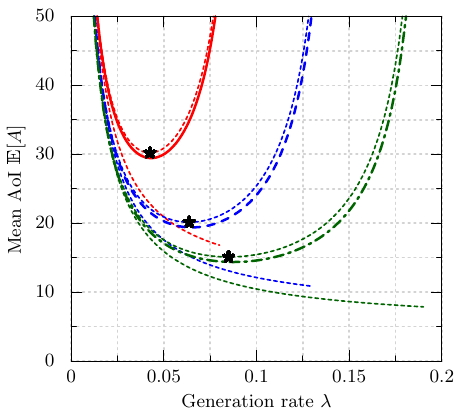}
\mbox{}
\\[-5ex]
(c) $s_H = 1$
\end{minipage}
\begin{minipage}[t]{0.49\textwidth}
\centering
\includegraphics[scale=1.0]{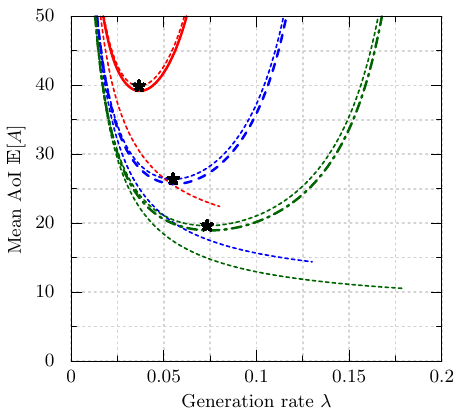}
\mbox{}
\\[-5ex]
(d) $s_H = 1.5$
\end{minipage}
 \caption{\label{Fig_E2}The mean AoI  as a function of the coefficient of the \yoshiaki{generation} rate $\lambda$, \yoshiaki{for different values of the coefficient of variation $s_H$ of the service times}. The explicit bounds \michel{of} Theorem \ref{theorem:A-bound-mean} have also been plotted. \yoshiaki{The $\star$ marks in the figure represent \michel{the points} ($\lambda^{\star}$, $\E_{\lambda^{\star}}^+[A]$) given by the closed-form formulas (\ref{eq:optimallambda}) and (\ref{eq:EA-star}), \michel{respectively}.}}
\end{figure}

\begin{table}[tp]
\centering
\yoshiaki{
\caption{\label{tab:EA}\yoshiaki{Comparison of $\lambda^{\star}$ given by (\ref{eq:optimallambda}) with the optimal $\lambda$ (denoted by $\lambda_{\mathrm{opt}}$) that minimizes $\E[A]$. The column ``$\E[A]$ for $\lambda_{\rm opt}$'' provides the value of $\E[A]$ at $\lambda = \lambda_{\rm opt}$, while the column ``$\E[A]$ for $\lambda^\star$'' provides the value of $\E[A]$ at $\lambda = \lambda^{\star}$.}}
\begin{tabular}{c|c||c|c|c|c|c}
$s_H$ & $\lambda_{\bg}$ & $\lambda_{\mathrm{opt}}$ & $\lambda^{\star}$ 
& $\E[A]$ for $\lambda_{\mathrm{opt}}$
& $\E[A]$ for $\lambda^{\star}$
& Difference in $\E[A]$
\\  \hline
0 & 0.8 & 0.108 & 0.102 & 10.109 & 10.137 & 0.028
\\
0 & 0.85 & 0.080 & 0.077 & 13.571 & 13.591 & 0.020
\\
0 & 0.9 & 0.052 & 0.051 & 20.509 & 20.522 & 0.013
\\
0.5 & 0.8 & 0.101 & 0.097 & 11.261 & 11.282 & 0.021
\\
0.5& 0.85 & 0.075 & 0.072 & 15.138 & 15.153 & 0.015
\\
0.5 & 0.9 & 0.049 & 0.048 & 22.906 & 22.915 & 0.009
\\
1.0 & 0.8 & 0.088 & 0.085 & 14.371 & 14.381 & 0.010
\\
1.0 & 0.85 & 0.065 & 0.064 & 19.391 & 19.398 & 0.007
\\
1.0 & 0.9 & 0.043 & 0.042 & 29.442 & 29.447 & 0.005
\\
1.5 & 0.8 & 0.075 & 0.073 & 18.959 & 18.964 & 0.005
\\
1.5 & 0.85 & 0.056 & 0.055 & 25.697 & 25.701 & 0.004
\\
1.5 & 0.9 & 0.037 & 0.037 & 39.185 & 39.188 & 0.003
\end{tabular}
}
\end{table}

\begin{figure}
\centering
\includegraphics[scale=1.0]{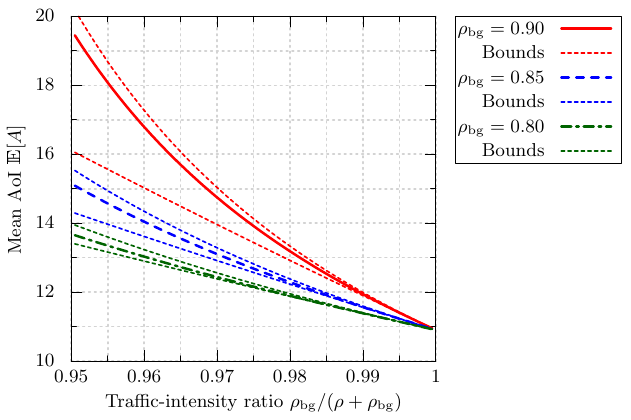}
 \caption{\label{Fig_E3}The mean AoI, as a function of the ratio of the traffic intensity $\rho_{\bg}/(\rho+\rho_{\bg})$, where $\lambda_{\bg}$ and $\mu$ are scaled up simultaneously so that $\rho_{\bg}$ is kept fixed. The explicit bounds in Theorem \ref{theorem:A-bound-mean} are also plotted.}
\end{figure}

\section{Discussion and concluding remarks}\label{sec:conc}
This paper has described a methodology to analyze the AoI in the context of the highly general GI+M/GI+GI/1 model. In a first main result, Theorem \ref{theorem:A-by-D}, we provide a closed-form expression for its Laplace-Stieltjes transform, enabling the numerical evaluation of the AoI moments. The second main result (Theorem \ref{theorem:A-bound}) concerns stochastic lower and upper bounds on the AoI, leading to explicit, insightful lower and upper bounds on the mean AoI (Theorem \ref{theorem:A-bound-mean}). We then approximate the tagged stream's inter-generation times via a phase-type distribution, which can be done at any desired precision level. For the resulting PH+M/GI+GI/1 model we succeed, as a third main result, in developing a stable computational algorithm for the mean AoI; see Theorem \ref{thm_meanA} for the resulting expression. \michel{We further observed that the generation rate that minimizes the upper bound in Theorem \ref{theorem:A-bound-mean}, as given in \eqref{eq:optimallambda}, yields a highly accurate proxy for the AoI minimizing generation rate.}

While with our GI+M/GI+GI/1 model we have reached a great level of generality, some further extensions can be thought of. In our setup we have, to model the background packet arrival process, relied on the fact that the superposition of many relatively homogeneous streams behaves essentially Poissonian. This means that we have to develop alternative techniques for settings in which the background stream does not correspond to a large traffic aggregate. A relevant specific case could be the one in which the background stream is modeled as the superposition of relative few streams with deterministic inter-arrival times, each of them with a random `phasing' (cf.\ the $N*$D/D/1 queue \cite[Section 15.2]{Roberts1996}). 

A second challenge could lie in the algorithmic computation of the AoI moments. Concrete objectives could concern optimizing the precision and run time of the algorithm for the mean AoI, as presented in Figure~\ref{figure:algorithm}, and its extension to higher moments. 

\appendices

\yoshiaki{
\section{Proof of Lemma \ref{lemma:EA-formula-PH}}
\label{appendix:EA-formula-PH}

We rewrite $f_G(x)$ as
\begin{align}
f_G(x) 
= 
\bm{\gamma} e^{-\theta x} \exp[(\theta\bm{I}+\bm{\Gamma}) x] (-\bm{\Gamma})\bm{e}
&=
\bm{\gamma} e^{-\theta x} 
\sum_{k=0}^{\infty} \frac{[(\theta\bm{I}+\bm{\Gamma}) x]^k}{k!}
\cdot
(-\bm{\Gamma})\bm{e}
\nonumber
\\
&=
\sum_{k=0}^{\infty} 
\frac{e^{-\theta x} (\theta x)^k}{k!}
\cdot
c_k.
\label{eq:f_G-uniform}
\end{align}
With this expression, we rewrite (\ref{eq:EA-by-q(y)-0}) as
\begin{equation}
\E[A]
=
\frac{\E[H]}{\mu}
+
\frac{\E[G^2]}{2\E[G]}
+
\sum_{k=0}^{\infty}
\frac{c_k}{\theta\E[G]}
\int_{y=0}^{\infty}
q(y) 
\cdot y \cdot \frac{e^{-\theta y}(\theta y)^k}{k!} \cdot \theta\,
\dd y.
\label{eq:EA-by-q(y)}
\end{equation}
From (\ref{eq:f_G-uniform}), we have 
\begin{align}
\E[G^2]
=
\int_0^{\infty}
y^2
\sum_{k=0}^{\infty}
\frac{e^{-\theta y}(\theta y)^k}{k!} 
\cdot
c_k
\, \dd y
&=
\sum_{k=0}^{\infty}
\frac{c_k}{\theta}
\int_{y=0}^{\infty} y^2 \cdot \frac{e^{-\theta y}(\theta y)^k}{k!} 
\cdot \theta\, \dd y
\nonumber
\\
&=
\sum_{k=0}^{\infty}
\frac{c_k}{\theta}
\cdot
\frac{(k+1)(k+2)}{\theta^2}.
\label{eq:G2-PH}
\end{align}
where the last equality is an immediate consequence of the definition
of the gamma-integral.
Also,
\begin{equation}
\int_0^{\infty}
q(y) 
\cdot y \cdot \frac{e^{-\theta y}(\theta y)^k}{k!} \cdot \theta\,
\dd y
=
(k+1)\int_0^{\infty}
q(y) 
\cdot \frac{e^{-\theta y}(\theta y)^{k+1}}{(k+1)!} 
\dd y
= 
(k+1) q_{k+1}.
\label{eq:qyy-int}
\end{equation}
We then obtain (\ref{eq:EA-formula-PH}) \michel{by combining} (\ref{eq:EA-by-q(y)}),
(\ref{eq:G2-PH}), and (\ref{eq:qyy-int}).
\qed
}

\section{Proof of Lemma \ref{lemma:b-coeff}}
\label{appendix:b-coeff}

Using the integral representation of the M/GI/1 busy period
length \cite[P.\ 653]{Cohen1982}, we have
\begin{align*}
f_{B_{\bg}}^*(\theta - \theta z)
&=
\int_0^{\infty} 
e^{-(\theta-\theta z)x}
\sum_{k=1}^{\infty} 
\frac{e^{-\lambda_{\bg} x} (\lambda_{\bg} x)^{k-1}}{k!}
\dd F_{Y_k}(x),
\end{align*}
where $Y_n$ denotes the sum of $k$ i.i.d random variables following
the CDF $F_{H_{\bg}/\mu}(\cdot)$. Therefore, 
\begin{align*}
f_{B_{\bg}}^*(\theta - \theta z)
&=
\sum_{m=0}^{\infty}
z^m
\sum_{k=1}^{\infty} 
\int_0^{\infty} 
\frac{e^{-\theta x}(\theta x)^m}{m!} 
\cdot
\frac{e^{-\lambda_{\bg} x} (\lambda_{\bg} x)^{k-1}}{k!}
\dd F_{Y_k}(x)
\\
&=
\sum_{m=0}^{\infty}
z^m
\sum_{k=1}^{\infty} 
\frac{1}{k}
\cdot
\frac{(m+k-1)!}{m!(k-1)!}
\cdot
\left(\frac{\theta}{\theta + \lambda_{\bg}}\right)^m
\left(\frac{\lambda_{\bg}}{\theta + \lambda_{\bg}}\right)^{k-1}
\\
&\hspace{10em}{}\cdot
\int_0^{\infty} 
\frac{e^{-(\theta+\lambda_{\bg}) x}\{(\theta + \lambda_{\bg})
x\}^{m+k-1}}{(m+k-1)!}
\dd F_{Y_k}(x)
\\
&=
\sum_{m=0}^{\infty}
z^m
\sum_{k=1}^{\infty} 
\frac{1}{k}
\comb{m+k-1}{m}
\left(\frac{\theta}{\theta + \lambda_{\bg}}\right)^m
\left(\frac{\lambda_{\bg}}{\theta + \lambda_{\bg}}\right)^{k-1}
y_k^{(m+k-1)}
\\
&=
\sum_{m=0}^{\infty}
z^m
\sum_{k=m}^{\infty} 
\frac{1}{k-m+1}
\comb{k}{m}
\left(\frac{\theta}{\theta + \lambda_{\bg}}\right)^m
\left(\frac{\lambda_{\bg}}{\theta + \lambda_{\bg}}\right)^{k-m}
y_{k-m+1}^{(k)},
\end{align*}
where
\[
y_k^{(m)} = \int_0^{\infty} 
\frac{e^{-\zeta x}(\zeta x)^m}{m!}
\dd F_{Y_k}(x),
\quad
m = 0,1,\ldots, k= 0,1,\ldots.
\]
By definition, it is readily verified that $y_k^{(m)}$ satisfies
the recursion (\ref{eq:y_k-recursion-1}) and (\ref{eq:y_k-recursion-k}).
We thus obtain (\ref{eq:b^i-def}) from (\ref{eq:b(z)-def}).
\qed

\yoshiaki{
\section{Proof of Lemma \ref{lemma:d-m}}
\label{appendix:d-m}

\subsection{Preliminaries}

Recall that $V(t)$ denotes the workload in the system, and that
$\hat{V}(t)/\mu$ denotes the virtual waiting time. 
Let $\hat{V}$ and $S$ denote generic random variables following the
(joint) stationary distribution of $(\hat{V}(t), S(t))$. 
Let $\hat{\bm{v}}(x)$ (for $x \geq 0$) denote a $1 \times M$ vector whose
$j$th ($j = 1,2,\ldots,M$) element represents the stationary joint
probability that the virtual waiting time is not greater than $x$ and the phase 
equals $j$:
\[
[\hat{\bm{v}}(x)]_j = \Pr(\hat{V} \leq x, S=j).
\]
We define $\bm{D}^*(\cdot)$ and $\hat{\bm{v}}^*(\cdot)$ as the LSTs of 
$\bm{D}(\cdot)$ and $\hat{\bm{v}}(\cdot)$:
\begin{align*}
\bm{D}^*(s) &= \int_0^{\infty} e^{-sx} \dd \bm{D}(x),
\quad
s \geq 0,
\\
\hat{\bm{v}}^*(s) &= \int_0^{\infty} e^{-sx} \dd \hat{\bm{v}}(x),
\quad
s \geq 0.
\end{align*}
We further define and $\hat{\bm{D}}(\cdot)$ 
and $\hat{\bm{D}}^*(\cdot)$ as
\begin{align}
\hat{\bm{D}}(x) &= \bm{D}(\mu x),
\quad
x \geq 0,
\label{eq:hat-D-def}
\\
\hat{\bm{D}}^*(s) 
&= 
\int_0^{\infty} e^{-sx} \dd \hat{\bm{D}}(x) = \bm{D}^*(s/\mu),
\quad
s \geq 0,
\nonumber
\end{align}
which correspond to $\bm{D}(\cdot)$ and $\bm{D}^*(\cdot)$ with
$H$ and $H_{\bg}$ replaced with $H/\mu$ and $H_{\bg}/\mu$.
We can verify that $\bm{D}^{(m)}$ defined in (\ref{eq:hatD-m}) satisfies
\begin{align*}
\hat{\bm{D}}^{(m)} 
&= 
(-\bm{\Gamma}\bm{e})\bm{\gamma}
\int_0^{\infty} 
\frac{e^{-\zeta x/\mu}(\zeta x/\mu)^m }{m!}
\dd F_H(x)
+
\bm{I}
\lambda_{\bg} \int_0^{\infty} 
\frac{e^{-\zeta x/\mu}(\zeta x/\mu)^m }{m!}
\dd F_{H_{\bg}}(x)
\\
&=
\int_0^{\infty} 
\frac{e^{-\zeta x}(\zeta x)^m }{m!}
\dd \hat{\bm{D}}(x), 
\end{align*}
where the second equality follows from the definitions
(\ref{eq:C-D(x)-def}) and (\ref{eq:hat-D-def}) of $\bm{D}(\cdot)$ and
$\hat{\bm{D}}(\cdot)$. 

Notice that the recursion (\ref{eq:Q-recursion-uniform}) for 
$\bm{Q}^{(n)}$ is equivalent to 
\begin{equation}
\bm{Q}^{(0)} = \bm{C},
\quad
\bm{Q}^{(n)} 
= 
\bm{C} + \int_0^{\infty} \dd \hat{\bm{D}}(x) \exp[\bm{Q}^{(n-1)}x],
\;\;
n = 1,2,\ldots,
\label{eq:Q-recursion}
\end{equation}
which can be verified with
\begin{align}
\exp[\bm{Q}^{(n)} x] &= 
\sum_{k=0}^{\infty} \frac{(\bm{Q}^{(n)} x)^k}{k!}
=
\sum_{k=0}^{\infty} \frac{e^{-\zeta x} (\zeta x)^k}{k!} 
\cdot 
(\bm{I}+\zeta^{-1}\bm{Q}^{(n)})^k.
\label{eq:uniformization-Q-n}
\end{align}
The expression (\ref{eq:uniformization-Q-n}) is known as the uniformization
technique \cite[P. 154]{Tijms1994}, which enables us to compute the
matrix exponential in a numerically stable way.

The following results are known for this model \cite{Takine1994,Takine2002}.

\begin{lemma}[\!\!\cite{Takine2002}]
\label{lemma:v-by-R}

The sequence $(\bm{Q}^{(n)})_{n=0,1,\ldots}$ given by (\ref{eq:Q-recursion})
is an elementwise non-decreasing sequence and its limit $\bm{Q} :=
\lim_{n \to \infty} \bm{Q}^{(n)}$ represents a proper transition rate
matrix of an irreducible continuous-time Markov chain with $M$ states.
Let $\bm{\kappa}$ denote the stationary probability vector associated with $\bm{Q}$:
\begin{equation}
\bm{\kappa}\bm{Q} = \bm{0},
\quad
\bm{\kappa}\bm{e} = 1.
\label{eq:kappa-def-2}
\end{equation}
The vector LST $\hat{\bm{v}}^*(s)$ of the stationary virtual waiting
time is then given by
\begin{equation}
\hat{\bm{v}}^*(s) =
(1-\rho-\rho_{\bg})\bm{\kappa}[\bm{I} - \bm{R}^*(s)]^{-1},
\quad
s \geq 0,
\label{eq:v-by-R}
\end{equation}
where $\bm{R}^*(s)$ is defined as
\begin{align}
\bm{R}^*(s)
=
\int_{x=0}^{\infty}e^{-sx} \dd x \int_{y=x}^{\infty}
\dd \hat{\bm{D}}(y) \exp[\bm{Q}(y-x)],
\label{eq:Rs-def}
\end{align}
and for $s \geq 0$, eigenvalues of $\bm{R}^*(s)$ have absolute values strictly less
than one, i.e., \begin{equation}[\bm{I}-\bm{R}^*(s)]^{-1} = \sum_{m=0}^{\infty}
(\bm{R}^*(s))^m < \infty\end{equation} holds.
\end{lemma}
\begin{remark}[\!\!\cite{Takine1994}]
For this model, we also have a more straightforward generalized version of the
Pollaczek–Khinchine formula, as follows:
\begin{equation}
\hat{\bm{v}}^*(s)[s\bm{I}+\bm{C}+\hat{\bm{D}}^*(s)] 
=
(1-\rho-\rho_{\bg})s \bm{\kappa}.
\label{eq:v-by-CD}
\end{equation}
This equation is readily deduced from (\ref{eq:kappa-def-2}) and
(\ref{eq:v-by-R}), by noting that (\ref{eq:Rs-def}) implies
\begin{equation}
[\bm{I}-\bm{R}^*(s)](s\bm{I}+\bm{Q}) =
s\bm{I}+\bm{C}+\hat{\bm{D}}^*(s).
\label{eq:R-and-CD}
\end{equation}
However, the advantage of the expression (\ref{eq:v-by-R}) over
(\ref{eq:v-by-CD}) is that $[\bm{I}-\bm{R}^*(s)]^{-1}$ exists for
all $s \geq 0$ as stated in Lemma \ref{lemma:v-by-R}.
This is not the case for $s\bm{I}+\bm{C}+\hat{\bm{D}}^*(s)$ because 
(\ref{eq:R-and-CD}) implies that 
$\det(s\bm{I}+\bm{C}+\hat{\bm{D}}^*(s)) = 0$ when $s$ equals an
eigenvalue of $-\bm{Q}$, which all have non-negative real parts owing
to the Perron–Frobenius theorem.
\end{remark}

\subsection{Proof of Lemma \ref{lemma:d-m}}

The following lemma is essential in proving Lemma \ref{lemma:d-m}:

\begin{lemma}
\label{lemma:hat-v-m-recursion}

Let $\hat{\bm{v}}^{(m)}$ (for $m=0,1,\ldots$) denote the $m$th
coefficient of $\hat{\bm{v}}^*(\zeta-\zeta z)$ as a function of
$z$: 
\[
\hat{\bm{v}}^*(\zeta-\zeta z) = \sum_{m=0}^{\infty}\hat{\bm{v}}^{(m)}
z^m.
\]
$\hat{\bm{v}}^{(m)}$ ($m=0,1,\ldots$) is determined by the following recursion:
\begin{align*}
\hat{\bm{v}}^{(0)} 
&= (1-\rho-\rho_{\bg})\bm{\kappa}
(\bm{I}-\bm{R}^{(0)})^{-1},
\\
\hat{\bm{v}}^{(m)} 
&= \left(\sum_{k=0}^{m-1} \hat{\bm{v}}^{(k)}\bm{R}^{(m-k)}\right)
(\bm{I}-\bm{R}^{(0)})^{-1},
\quad
m = 1,2,\ldots,
\end{align*}
where $\bm{R}^{(m)}$ (for $m=0,1,\ldots$) is given by (\ref{eq:R-uniform}).
\end{lemma}
\begin{proof}
From (\ref{eq:Rs-def}), we have
\begin{align*}
\bm{R}^*(\zeta - \zeta z)
&=
\int_{x=0}^{\infty} e^{-(\zeta - \zeta z)x} \dd x
\int_{y=x}^{\infty} \dd \hat{\bm{D}}(y) \exp[\bm{Q}(y-x)]
\\
&=
\int_{y=0}^{\infty} 
\dd \hat{\bm{D}}(y) 
\int_{x=0}^y
\exp[\bm{Q}(y-x)]
e^{-(\zeta - \zeta z)x} \dd x
\\
&=
\int_{y=0}^{\infty} 
\dd \hat{\bm{D}}(y) 
\int_{x=0}^y
\sum_{k=0}^{\infty}
\frac{e^{-\zeta (y-x)}\{\zeta (y-x)\}^k}{k!}
\cdot
(\bm{I}+\zeta^{-1}\bm{Q})^k
e^{-(\zeta - \zeta z)x} \dd x
\\
&=
\sum_{k=0}^{\infty}
\int_{y=0}^{\infty} 
e^{-\zeta y}
\dd \hat{\bm{D}}(y) 
(\bm{I}+\zeta^{-1}\bm{Q})^k
\int_{x=0}^y
\frac{\{\zeta (y-x)\}^k}{k!}
e^{\zeta zx} \dd x
\\
&=
\sum_{k=0}^{\infty}
\int_{y=0}^{\infty} 
e^{-(\zeta - \zeta z) y}
\dd \hat{\bm{D}}(y) 
(\bm{I}+\zeta^{-1}\bm{Q})^k
\int_{x=0}^y
\frac{\{\zeta (y-x)\}^k}{k!}
e^{-\zeta z(y-x)} \dd x
\\
&=
\sum_{k=0}^{\infty}
z^{-(k+1)}\zeta^{-1}
\int_{y=0}^{\infty} 
e^{-(\zeta - \zeta z) y}
\dd \hat{\bm{D}}(y) 
(\bm{I}+\zeta^{-1}\bm{Q})^k
\int_{x=0}^y
\frac{e^{-\zeta z(y-x)} \{\zeta z(y-x)\}^k }{k!}
\cdot
\zeta z
\dd x
\\
&=
\sum_{k=0}^{\infty}
\sum_{m=k+1}^{\infty}
z^{-(k+1)}\zeta^{-1}
\int_{y=0}^{\infty} 
e^{-(\zeta - \zeta z) y}
\cdot
\frac{e^{-\zeta zy} (\zeta zy)^m }{m!}
\dd \hat{\bm{D}}(y) 
(\bm{I}+\zeta^{-1}\bm{Q})^k
\\
&=
\sum_{k=0}^{\infty}
\sum_{m=k+1}^{\infty}
z^{m-k-1}\zeta^{-1}
\int_{y=0}^{\infty} 
\frac{e^{-\zeta y}(\zeta y)^m }{m!}
\dd \hat{\bm{D}}(y) 
(\bm{I}+\zeta^{-1}\bm{Q})^k
\\
&=
\sum_{k=0}^{\infty}
\sum_{m=k+1}^{\infty}
z^{m-k-1}\zeta^{-1}
\hat{\bm{D}}^{(m)}
(\bm{I}+\zeta^{-1}\bm{Q})^k
\\
&=
\sum_{k=0}^{\infty}
\sum_{m=0}^{\infty}
z^m\zeta^{-1}
\hat{\bm{D}}^{(m+k+1)}
(\bm{I}+\zeta^{-1}\bm{Q})^k
\\
&=
\sum_{m=0}^{\infty}
z^m
\zeta^{-1}
\sum_{k=0}^{\infty}
\hat{\bm{D}}^{(m+k+1)}
(\bm{I}+\zeta^{-1}\bm{Q})^k,
\end{align*}
Therefore, we have proved
\[
\bm{R}^*(\zeta-\zeta z)
=
\sum_{m=0}^{\infty} \bm{R}^{(m)} z^m.
\]
Note that we have from (\ref{eq:v-by-R}), 
\begin{align*}
\hat{\bm{v}}(\zeta-\zeta z)
&=
(1-\rho-\rho_{\bg})\bm{\kappa}
+
\hat{\bm{v}}(\zeta-\zeta z)
\bm{R}^*(\zeta-\zeta z),
\end{align*}
i.e.,
\begin{align*}
\sum_{m=0}^{\infty}
\hat{\bm{v}}^{(m)} z^m
&=
(1-\rho-\rho_{\bg})\bm{\kappa}
+
\sum_{m=0}^{\infty}
z^m
\sum_{k=0}^m
\hat{\bm{v}}^{(k)} \bm{R}^{(m-k)}.
\end{align*}
Comparing the coefficients of both sides of this equation and
re-arranging terms, we have established the claim of Lemma \ref{lemma:hat-v-m-recursion}.
\end{proof}
\begin{remark}
\label{remark:IR-inverse_IQ_subst_proof}

The statements in Remark \ref{remark:IR-inverse_IQ_subst} are verified
as follows. 
\begin{itemize}
\item[(i)] The existence of $(\bm{I}-\bm{R}^{(0)})^{-1}$ is
ensured by Lemma \ref{lemma:v-by-R} because $\bm{R}^{(0)} = \bm{R}^*(\zeta)$.
\item[(ii)] From (\ref{eq:C-D(x)-def}) and (\ref{eq:theta-zeta}), 
we have $\zeta = \max_i |[\bm{C}]_{i,i}|$, i.e., it equals to the maximum
absolute value  of diagonal elements of $\bm{C}$.
As shown in Lemma \ref{lemma:v-by-R}, $\bm{Q}^{(n)}$ is an elementwise
non-decreasing sequence with $\bm{Q}^{(0)}=\bm{C}$, so that we have
$\max_i |[\bm{Q}^{(n)}]_{i,i}| \leq \zeta$. Therefore, 
$\bm{I}+\zeta^{-1}\bm{Q}^{(n)}$ is a substochastic matrix.
\end{itemize}
\end{remark}

From the conditional PASTA property \cite{Door1988}, the LST
$f_D^*(s)$ of the stationary system delay $D$ is obtained as
\begin{equation}
f_D^*(s) 
= 
\frac{\hat{\bm{v}}^*(s)(-\bm{\Gamma}\bm{e})f_H^*(s/\mu)}{\lambda}.
\label{eq:f_D-PH}
\end{equation}
Therefore, we obtain Lemma \ref{lemma:d-m} from 
Lemma \ref{lemma:hat-v-m-recursion} and (\ref{eq:f_D-PH}).
\qed
}

\yoshiaki{
\section{Proof of Theorem \ref{thm_meanA}}
\label{appendix:meanA-PH}

By definition (\ref{eq:b_ell-recursion}), it is obvious that
\[
\{b(z)\}^{\ell} = \sum_{m=0}^{\infty} b_{\ell}^{(m)} z^m.
\]
We then have from (\ref{eq:d-m-def}), (\ref{eq:b(z)-def}), and (\ref{eq:psi-by-b}),
\begin{align}
\frac{f_D^*(\psi(\theta - \theta z))}{\psi(\theta - \theta z)}
=
\frac{f_D^*(\zeta - \zeta b(z))}{\zeta - \zeta b(z)}
\nonumber
&=
\frac{1}{\zeta}
\sum_{k=0}^{\infty}\{b(z)\}^k
\sum_{m=0}^{\infty} d^{(m)} \{b(z)\}^m
\nonumber
\\
&=
\frac{1}{\zeta}
\sum_{\ell=0}^{\infty} \{b(z)\}^{\ell}
\sum_{m=0}^{\ell}d^{(m)}
\nonumber
\\
&=
\sum_{i=0}^{\infty}
z^i
\cdot
\frac{1}{\zeta}
\sum_{\ell=0}^{\infty} 
b_{\ell}^{(i)}\sum_{m=0}^{\ell}d^{(m)},
\label{eq:psi-frac}
\end{align}
where we put $\ell = m+k$ to obtain the second equality.
Therefore, we obtain from (\ref{eq:q^*-by-psi}) and (\ref{eq:psi-frac}),
\[
q^*(z) 
=
\sum_{k=0}^{\infty}
z^k
\left(
-\frac{(k+1)(1-\rho_{\bg})}{\theta^2}
+\frac{\E[D]}{\theta}
+\frac{1}{\theta}
\sum_{i=0}^k 
\frac{1}{\zeta}
\sum_{\ell=0}^{\infty} 
b_{\ell}^{(i)}\sum_{m=0}^{\ell}d^{(m)}
\right),
\]
i.e.,
\[
q_k 
=
-\frac{(k+1)(1-\rho_{\bg})}{\theta^2}
+\frac{\E[D]}{\theta}
+\frac{1}{\theta}
\sum_{i=0}^k 
\frac{1}{\zeta}
\sum_{\ell=0}^{\infty} 
b_{\ell}^{(i)}\sum_{m=0}^{\ell}d^{(m)},
\quad
k = 0,1,\ldots.
\]
It is then readily verified that $(q_k)_{k=0,1,\ldots}$ satisfies the
recursion (\ref{eq:q-recursion-0}) and (\ref{eq:q-recursion-k}).
\qed
}

\end{document}